\documentclass{lmcs}


\usepackage{hyperref}
\usepackage{amsmath}
\usepackage{amssymb}
\usepackage{booktabs}
\usepackage{placeins}
\usepackage{tikz}
\usetikzlibrary{automata,positioning,arrows.meta}

\newcommand{\appref}[1]{Appendix~\ref{#1}}

\newcommand{\ATL}{\mathrm{ATL}}
\newcommand{\ATLs}{\mathrm{ATL}^{*}}
\newcommand{\ATLmp}{\mathrm{ATL}_{\mathsf{mp}}}
\newcommand{\ATLsmp}{\mathrm{ATL}^{*}_{\mathsf{mp}}}
\newcommand{\LTL}{\mathrm{LTL}}
\newcommand{\GR}{\mathrm{GR}(1)}
\newcommand{\ATLsmpGR}{\ATLsmp[\GR]}

\newcommand{\game}{\mathcal{G}}
\newcommand{\arena}{\mathcal{A}}
\newcommand{\St}{\mathit{St}}
\newcommand{\Ac}{\mathit{Ac}}
\newcommand{\Av}{\mathit{Av}}
\newcommand{\acv}{\vec{ac}}
\newcommand{\tr}{\mathit{tr}}
\newcommand{\lab}{\mathit{lab}}
\newcommand{\AP}{\mathit{AP}}
\newcommand{\sinit}{s_{\mathit{init}}}
\newcommand{\pay}{\mathit{pay}}
\newcommand{\mpf}{\mathit{mp}}

\newcommand{\Core}{\mathrm{Core}}
\newcommand{\coal}[1]{\langle\!\langle #1 \rangle\!\rangle}
\newcommand{\dcoal}[1]{[\![ #1 ]\!]}
\newcommand{\ext}[1]{[\![#1]\!]}
\newcommand{\out}{\pi}
\newcommand{\proj}{\mathit{pr}}
\newcommand{\pre}{\mathit{pre}}
\newcommand{\pri}{\mathit{pri}}
\newcommand{\dpa}{\mathcal{D}}
\newcommand{\dm}{\mathit{dim}}
\newcommand{\PR}{\mathrm{pr}}
\newcommand{\FM}{\mathrm{fm}}
\newcommand{\ML}{\mathrm{ml}}
\newcommand{\sub}{\mathit{sub}}
\newcommand{\size}[1]{|#1|}
\newcommand{\Xop}{\mathbf{X}}
\newcommand{\Uop}{\mathbin{\mathbf{U}}}
\newcommand{\Fop}{\mathbf{F}}
\newcommand{\Gop}{\mathbf{G}}
\newcommand{\Dominated}{\textsc{Dominated}}

\newcommand{\MPG}{\textsc{Mpg}}
\newcommand{\MPP}{\textsc{Mpp}}
\newcommand{\NP}{\textsc{NP}}
\newcommand{\coNP}{\textsc{coNP}}
\newcommand{\PSPACE}{\textsc{PSpace}}
\newcommand{\PTIME}{\mathrm{P}}
\newcommand{\TwoExp}{2\textsc{Exptime}}

\newcommand{\hW}{\widehat{W}}

\begin{document}

\title[ATL with Mean-Payoff Guarantees]{Alternating-Time Temporal
  Logic\texorpdfstring{\\}{ }with Mean-Payoff Guarantees}

\author[M.~Najib]{Muhammad Najib\lmcsorcid{0000-0002-6289-5124}}
\address{Heriot-Watt University, Edinburgh, UK}
\email{m.najib@hw.ac.uk}

\begin{abstract}
  Alternating-time temporal logic and its extensions provide several ways of combining strategic and quantitative reasoning. We study a particular combination: whether a coalition has a single strategy that enforces a temporal objective while guaranteeing given long-run mean-payoff thresholds. We introduce $\ATLsmp$, an extension of $\ATLs$ over weighted concurrent game structures in which each strategic modality carries a conjunctive mean-payoff constraint. The temporal and quantitative requirements must hold against every behaviour of the remaining agents, and the existence of such a strategy cannot in general be reduced to the two requirements considered separately. For one-dimensional constraints, model checking is $\TwoExp$-complete under both perfect-recall and finite-memory semantics, matching $\ATLs$. For the pure quantitative fragment and fragments restricted to $\ATL$ or $\GR$ temporal objectives, model checking has lower complexity. With multi-dimensional conjunctive constraints, model checking under finite-memory semantics remains $\TwoExp$-complete. We show that memoryless, finite-memory, and perfect-recall abilities form a strict hierarchy, while finite-memory strategies still achieve every threshold strictly below the perfect-recall supremum. We give tight linear upper and lower bounds on the required memory as a function of the denominator of the threshold, even when the game and temporal monitor are fixed. We give several examples of properties expressible in the logic, including temporal synthesis with performance guarantees and aggregate and multi-criteria objectives. We also relate the logic to cooperative rational verification, showing that it can express beneficial deviations from fixed payoff baselines, but not directly reproduce the standard $\ATLs$ encoding of the core for dichotomous preferences.
\end{abstract}

\maketitle

\section{Introduction}\label{sec:intro}

Alternating-time temporal logic ($\ATL$) and its more expressive extension
$\ATLs$ \cite{AlurHK02} are among the most widely used formalisms for
reasoning about multi-agent systems. Their central construct is the
\emph{strategic modality} $\coal{C}\psi$, which states that the coalition $C$
has a strategy enforcing the temporal property $\psi$, whatever the remaining
agents do. This modality underlies much of the work on verification and
synthesis for open and multi-agent systems, as well as more expressive
languages such as Strategy Logic \cite{ChatterjeeHP10,MogaveroMPV14}.

$\ATL$ and $\ATLs$ cannot express the long-run performance associated with
a coalition strategy. The formula $\coal{C}\Gop\,\mathit{safe}$ states that
$C$ can keep the system safe, but says nothing about the energy consumed, the
throughput maintained or the reward obtained. Mean-payoff games
\cite{EhrenfeuchtM79,ZwickP96} express such long-run quantities, but do not
combine them with temporal requirements. We study whether a coalition has a
single strategy that enforces a temporal objective and a mean-payoff
threshold against every counter-strategy. This is stronger than requiring the
two objectives to be enforceable separately, as \autoref{prop:nondecomp}
shows.

We introduce $\ATLsmp$, a conservative extension of
$\ATLs$ interpreted over weighted concurrent game structures in which each
strategic modality carries a mean-payoff constraint,
\[
  \coal{C}_{\Lambda}\,\psi ,
  \qquad
  \Lambda ::= \top \mid \mpf_j \geq q \mid \Lambda \wedge \Lambda .
\]
Setting $\Lambda = \top$ recovers $\ATLs$ exactly, while setting $\psi = \top$
yields a query about the payoff vectors a coalition can enforce, in the sense
of multi-mean-payoff games \cite{VelnerCDHRR15,ChatterjeeRR14}. The combined
case has both a non-trivial temporal objective and a non-trivial mean-payoff
condition. For example,
\[
  \coal{\{1,2\}}_{\mpf_1 \geq 0\, \wedge\, \mpf_2 \geq 0}
    \bigl(\Gop(r_1 \to \Fop d_1) \wedge \Gop(r_2 \to \Fop d_2)\bigr) ,
\]
states that the two agents have a joint policy that responds to every request
of either type. The same policy must also keep both long-run rewards non-negative against every behaviour of the environment.

\paragraph{Contributions.}
Our contributions are as follows.

\begin{enumerate}[(1)]
\item \emph{Logic and semantics.} We define weighted concurrent game
  structures and the logic $\ATLsmp$ over them (\autoref{sec:model}), and
  show that it conservatively extends $\ATLs$. The combined modality cannot
  in general be reduced to the conjunction of an $\ATLs$ ability and a
  multi-mean-payoff ability:
  \[
    \coal{C}_{\Lambda}\psi \;\not\equiv\;
    \coal{C}\psi \wedge \coal{C}_{\Lambda}\top .
  \]
  Conjoining the corresponding separate queries therefore does not capture
  the requirement that a single coalition strategy enforce both objectives
  (\autoref{prop:nondecomp}).

\item \emph{Game characterisation.} The usual sequentialisation of a
  concurrent game introduces an intermediate state after each coalition move.
  This preserves mean payoff and stutter-invariant temporal objectives, but
  not arbitrary $\LTL$ formulae containing $\Xop$. We give a
  \emph{round-preserving} sequentialisation in which the temporal automaton
  advances once per concurrent round, and establish strategy correspondence
  under perfect-recall and finite-memory semantics (\autoref{thm:seq} and
  \autoref{thm:fm-correspondence}).

\item \emph{Model checking.} We give a bottom-up algorithm which, for each
  strategic subformula, builds the product of the round-preserving
  sequentialisation with a deterministic parity automaton for the path
  formula and solves the resulting mean-payoff parity game
  \cite{ChatterjeeHJ05,ChatterjeeD12}. With one-dimensional constraints,
  model checking is $\TwoExp$-complete under both perfect-recall and
  finite-memory semantics (\autoref{thm:main}), matching $\ATLs$
  \cite{AlurHK02}. We give sharper bounds for restricted temporal fragments
  in \autoref{sec:fragments}. We then show that model checking with arbitrary
  conjunctive constraints remains $\TwoExp$-complete under finite-memory
  semantics; the corresponding perfect-recall problem remains open
  (\autoref{sec:multi}).

\item \emph{Memory requirements.} We show that memoryless, finite-memory and
  perfect-recall abilities form a strict hierarchy on positive formulae
  (\autoref{thm:sep}), while every threshold strictly below the
  perfect-recall supremum is achievable with finite memory
  (\autoref{thm:approx}). The required memory may grow linearly in the
  denominator of the threshold, and hence exponentially in its binary
  encoding; \autoref{thm:memory-bounds} gives matching bounds.

\item \emph{Applications and expressive limits.} We give several examples
  of properties expressible in $\ATLsmp$, including temporal synthesis with
  performance guarantees and aggregate and multi-criteria objectives. We also
  relate the logic to cooperative rational verification: beneficial
  deviations from fixed payoff baselines are expressible, but the standard
  $\ATLs$ encoding of the core for dichotomous preferences does not extend
  directly to mean-payoff preferences, since the relevant deviation
  thresholds depend on the payoff of the candidate profile
  (\autoref{sec:apps}).
\end{enumerate}

\autoref{tab:main} summarises the principal model-checking bounds. More
refined bounds for restricted one-dimensional temporal fragments are given in
\autoref{tab:summary}. For a fixed number of dimensions, the quantitative,
$\ATLmp$ and $\ATLsmpGR$ fragments are moreover pseudo-polynomial under
finite-memory semantics (\autoref{prop:fixed-dim}).

\begin{table}[!ht]
  \centering
  \caption{Main model-checking complexity results, where ``-c'' abbreviates
    ``-complete''. The entries for the full logic concern arbitrary nesting. Under memoryless semantics the
    $\PSPACE$ bound holds for arbitrary conjunctive constraints, and
    hardness already holds when every constraint is $\top$. The
    multi-dimensional quantitative entries concern strategic atoms
    $\coal{C}_{\Lambda}\top$; under finite-memory and perfect-recall
    semantics, flat Boolean combinations of such atoms are in
    $\PTIME^{\NP}_{\parallel}$.}
  \label{tab:main}
  \small
  \setlength{\tabcolsep}{5pt}
  \begin{tabular}{lccc}
    \toprule
    & \multicolumn{2}{c}{Full $\ATLsmp$} & Quantitative atoms \\
    \cmidrule(lr){2-3}\cmidrule(lr){4-4}
    Semantics & one-dimensional & multi-dimensional & multi-dimensional \\
    \midrule
    Memoryless
      & \multicolumn{2}{c}{$\PSPACE$-c (Thm.~\ref{thm:memoryless})}
      & $\NP$-c (Prop.~\ref{prop:pure-multi}) \\
    Finite memory
      & $\TwoExp$-c (Thm.~\ref{thm:main})
      & $\TwoExp$-c (Thm.~\ref{thm:multi-full})
      & $\coNP$-c (Prop.~\ref{prop:pure-multi}) \\
    Perfect recall
      & $\TwoExp$-c (Thm.~\ref{thm:main})
      & open (Open Prob.~\ref{oprob:multi-pr})
      & $\coNP$-c (Prop.~\ref{prop:pure-multi}) \\
    \bottomrule
  \end{tabular}
\end{table}

\FloatBarrier

\paragraph{Related work.}
Bulling and Goranko \cite{BullingG13,BullingG22} study concurrent games with
accumulated utilities and utility-dependent action guards. They identify
arithmetic path constraints interpreted over mean payoffs as a direction for
future work \cite{BullingG22}. In contrast, the conditions considered here
are prefix-independent. Resource-bounded variants of $\ATL$
\cite{AlechinaLNR17,BullingF10,DellaMonicaNM11} constrain a finite budget
along a play. Energy constraints originate in weighted timed systems
\cite{BouyerFLMS08}; in graph games, energy objectives constrain the running
balance \cite{ChatterjeeD12,ChatterjeeDHR10}. Resource constraints and energy objectives both depend on
the preceding prefix. Boker et al.\ \cite{Boker0KKS14} study the related
setting of $\LTL$ with accumulative values and identify its decidability
boundary. At the algorithmic level, \autoref{lem:fm-energy} uses the
corresponding relationship between finite-memory mean-payoff parity
objectives and energy parity objectives.

Quantitative variants that assign a degree of satisfaction to a formula,
rather than evaluating a quantity along an outcome, have been studied for
Strategy Logic \cite{BouyerKMMMP23} and for $\ATLs$ \cite{MuranoNZ23}.
Murano, Neider and Zimmermann give a multi-valued semantics for the strategic
and temporal operators of $\ATL$ and $\ATLs$, designed to make satisfaction
robust to small violations of assumptions. Model checking
$\mathrm{rATL}^{*}$ is also $\TwoExp$-complete, but its quantitative
component measures degrees of satisfaction, whereas $\ATLsmp$ constrains the
mean payoff of an outcome. Their analysis
does not consider the distinctions between memoryless, finite-memory and
perfect-recall strategies studied here.

$\LTL$ synthesis with mean-payoff objectives has been studied in two-player
turn-based games \cite{BloemCHJ09,BohyBFJR13}. The problem reduces to
mean-payoff parity games and is $\TwoExp$-complete. For mean-payoff parity
games, exact optimality may require infinite memory, while finite-memory
strategies suffice for $\varepsilon$-approximation \cite{ChatterjeeHJ05};
this property carries over to $\LTL$ synthesis with mean-payoff objectives
\cite{BohyBFJR13}. Our setting allows arbitrary coalitions and
nested strategic modalities, and the underlying game is concurrent rather
than turn-based. \autoref{sec:seq} extends the corresponding game reduction
to concurrent games and arbitrary coalitions.

Rational verification asks which temporal properties hold in the stable outcomes of a multi-agent system \cite{GutierrezHW17,GutierrezNPW20,AbateGMPSSW21}; most of that work uses Nash equilibrium and its refinements, with equilibria in limit-average games studied in \cite{UmmelsW11} and
rational-verification problems for memoryless mean-payoff games with
$\omega$-regular specifications in \cite{GutierrezSW21}. We compare $\ATLsmp$ with the cooperative, core-based formulations of \cite{GutierrezKW23,GutierrezLNSW24} in \autoref{sec:rv}.

Quantitative languages \cite{ChatterjeeDH10} provide an automata-theoretic account of long-run values without strategic modalities. Strategy Logic \cite{ChatterjeeHP10,MogaveroMPV14} permits strategies to be quantified, named, and bound explicitly, but its standard form contains no payoff terms or mean-payoff objectives. Combining named strategy profiles with such payoff terms would allow the payoffs of a deviation to be compared with those of an enclosing profile, as required by the rational-verification application discussed in \autoref{sec:rv}.

\section{A motivating example}\label{sec:example}

The following example illustrates the combined temporal and quantitative
requirement, the resulting threshold trade-off, and the need for memory.

\begin{exa}[Coordinated warehouse robots]\label{exa:service}
  Two mobile robots, $1$ and $2$, jointly transport oversized pallets from
  two loading bays to a dispatch area. Transporting a pallet requires both
  robots to select the same loading bay, $L$ or $R$. The two bays feed
  different production lines, and both must be served repeatedly. An environment agent
  $e$ represents other warehouse traffic and determines whether the shared
  transport aisle is clear. The arena is shown in \autoref{fig:service}.

  At the staging state $m$ each robot chooses between $L$ and $R$, while the
  environment chooses between $\mathit{clear}$ and $\mathit{blocked}$. If the
  robots agree on $L$ and the aisle is clear, they complete a delivery from
  the left bay and the play moves to $\ell$; agreement on $R$ similarly
  completes a delivery from the right bay and leads to $r$. If the robots
  choose different bays, or if the aisle is blocked, the system enters the
  recovery state $c$. From $\ell$ and $r$ the robots return to the staging
  state. At $c$ the warehouse safety controller guarantees that the aisle is
  cleared before the next attempt, so the environment has no
  $\mathit{blocked}$ action there. Thus a blockage may delay a delivery but
  cannot prevent deliveries indefinitely.

  States are labelled over $\AP = \{d_1, d_2\}$, where $d_1$ and $d_2$
  record completion of a delivery from the left and right bays respectively:
  $\lab(\ell) = \{d_1\}$, $\lab(r) = \{d_2\}$ and
  $\lab(m) = \lab(c) = \emptyset$. The two weight dimensions measure the
  long-run service obtained by the two production lines. Completing a
  delivery contributes $4$ to the corresponding dimension, and entering
  recovery incurs a unit penalty in both:
  \[
    w(m) = (0,0), \quad
    w(\ell) = (4,0), \quad
    w(r) = (0,4), \quad
    w(c) = (-1,-1).
  \]
\end{exa}

\begin{figure}
  \centering
  \begin{tikzpicture}[
      ->, >={Stealth[round]}, shorten >=1pt, semithick,
      every state/.style={draw, minimum size=8.5mm, inner sep=1pt},
      lbl/.style={font=\footnotesize, inner sep=1.5pt, fill=white}]
    \node[state] (m) at (0,1.3)   {$m$};
    \node[state] (c) at (0,-1.3)  {$c$};
    \node[state] (l) at (4.0,1.3) {$\ell$};
    \node[state] (r) at (4.0,-1.3){$r$};

    \path
      (m) edge node[lbl,above]            {$(L,L,\mathit{clear})$} (l)
      (m) edge node[lbl,pos=.30]          {$(R,R,\mathit{clear})$} (r)
      (m) edge node[lbl,left,align=right]
            {$(*,*,\mathit{blocked})$\\ or disagreement}           (c)
      (c) edge node[lbl,pos=.80]          {$(L,L)$}             (l)
      (c) edge[bend right=42] node[lbl,below]            {$(R,R)$}             (r)
      (c) edge[loop below] node[lbl,below]{disagreement}        (c)
      (l) edge[bend right=42] node[lbl,above] {$*$}             (m)
      (r) edge[bend left=42]  node[lbl,below] {$*$}             (m);
  \end{tikzpicture}
  \caption{The warehouse arena of \autoref{exa:service}. An edge label
    $(a_1,a_2,a_e)$ records the actions of the two robots and the
    environment. The symbol $*$ matches every available robot action, and
    ``disagreement'' abbreviates the two profiles in which the robots choose
    different loading bays. At $c$ the aisle has been cleared, so the
    environment has only the action $\mathit{clear}$ and the outgoing edge
    labels specify only the robots' actions. Weights are $w(m)=(0,0)$, $w(\ell)=(4,0)$,
    $w(r)=(0,4)$ and $w(c)=(-1,-1)$.}
  \label{fig:service}
\end{figure}
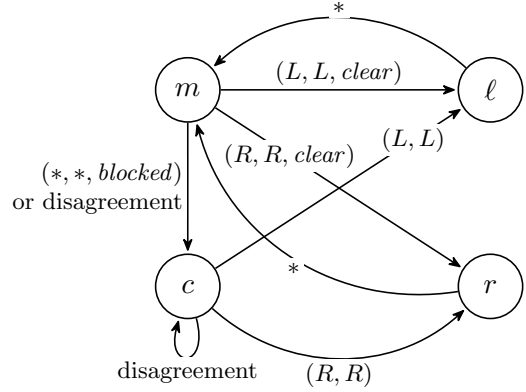

Write $C = \{1,2\}$ for the coalition of robots, and consider
\[
  \psi \;=\; \Gop\Fop d_1 \wedge \Gop\Fop d_2
  \qquad\text{and}\qquad
  \Lambda_q \;=\; \mpf_1 \geq q \wedge \mpf_2 \geq q .
\]
The temporal objective requires both production lines to be served infinitely
often, while $\Lambda_q$ places a lower bound on their long-run service
levels once recovery penalties are taken into account.

Alternating between the two loading bays satisfies $\psi$, while repeatedly
serving one bay can provide a high payoff in the corresponding dimension.
These properties may therefore be witnessed by different policies. The
formula $\coal{C}_{\Lambda_q}\psi$ instead requires one joint policy to
satisfy both, which by \autoref{prop:nondecomp} cannot in general be
decomposed into separate qualitative and quantitative abilities.

The achievable thresholds depend on how service is divided between the two
bays. Suppose the aisle is blocked whenever the environment is able to block
it, that is, at every visit to $m$. A joint policy that always agrees and
alternates between the two bays then produces the play
$(m\, c\, \ell\, m\, c\, r)^{\omega}$, whose weights sum over one cycle to
$(4,4) - (2,2) = (2,2)$ in six steps, so both mean payoffs equal
$\tfrac13$. Hence $\game, m \models \coal{C}_{\Lambda_{1/3}}\psi$. No larger
common threshold can be guaranteed against this environment. Every delivery
is then preceded by a visit to $c$, so each delivery contributes at most
$4 - 2 = 2$ to the sum of the two dimensions over at least three states, and
the long-run average of that sum is at most $\tfrac23$. Since
$\liminf_n a_n + \liminf_n b_n \leq \liminf_n (a_n + b_n)$, we get
$\mpf_1(\out) + \mpf_2(\out) \leq \tfrac23$ for every play $\out$, so both
mean payoffs exceeding $\tfrac13$ is impossible. Thus
$\game, m \not\models \coal{C}_{\Lambda_{q}}\psi$ for every
$q > \tfrac13$. More generally, the enforceable threshold
vectors are downward closed. Their Pareto boundary describes the trade-off
between the two service levels. As \autoref{sec:memory} shows, a boundary
threshold need not be attained by a finite-memory strategy.

No memoryless joint policy satisfies $\psi$. Such a policy fixes the robots'
choices at both $m$ and $c$. If they agree at $m$, the environment can always
play $\mathit{clear}$, so only the bay selected at $m$ is ever served. If
they disagree at $m$, every visit to $m$ leads to $c$, where the fixed joint
action completes at most one of the two deliveries. Either way the
environment can prevent one of $d_1$ and $d_2$ from recurring. A two-state
controller suffices: its memory records the bay served most recently.

\section{Models and logic}\label{sec:model}

Given any set $X$, we use $X^{*}$, $X^{\omega}$ and $X^{+}$ for,
respectively, the sets of finite, infinite, and non-empty finite sequences of
elements in $X$. For $Y \subseteq X$ we write $X_{-Y}$ for $X \setminus Y$
and $X_{-i}$ if $Y = \{i\}$. We extend this notation to tuples
$\vec{x} = (x_1, \dots, x_n)$, writing $\vec{x}_{C}$ for the restriction of
$\vec{x}$ to the indices in $C$ and $\vec{x}_{-C}$ for its restriction to the
complement. For a sequence $v$ we write $v[t]$ or $v^{t}$ for the element in
position $t+1$, so that $v[0] = v^{0}$ is the first element, and $v^{\geq k}$
for the suffix $v[k]v[k{+}1]\cdots$. Given two vectors
$\vec{a}, \vec{b} \in \mathbb{Q}^{d}$, the notations $\vec{a} \geq \vec{b}$
and $\vec{a} > \vec{b}$ denote componentwise inequality. Rational numbers are represented in the
usual binary encoding of numerator and denominator, and $\|q\|$ denotes the
size of that encoding.

\paragraph{Mean-Payoff.}
For an infinite sequence of real numbers
$r_0 r_1 r_2 \cdots \in \mathbb{R}^{\omega}$, its \emph{mean payoff} is
\[
  \mpf(r) \;=\; \liminf_{n \to \infty}\frac{1}{n}\sum_{t<n} r_t .
\]
We use the lower long-run average throughout, as in the mean-payoff game
models considered below \cite{ZwickP96,ChatterjeeHJ05,VelnerCDHRR15}. Mean
payoff is prefix-independent: $\mpf(r) = \mpf(r^{\geq k})$ for every
$k \geq 0$.

\paragraph{Temporal Logics.}
We use $\LTL$ \cite{Pnueli77} with the usual temporal operators $\Xop$
(``next'') and $\Uop$ (``until''), and the derived operators $\Gop$
(``always'') and $\Fop$ (``eventually''). We also use $\GR$
\cite{BloemJPPS12}, the fragment of $\LTL$ consisting of formulae of the form
$(\Gop\Fop\psi_1 \wedge \cdots \wedge \Gop\Fop\psi_m) \to
 (\Gop\Fop\vartheta_1 \wedge \cdots \wedge \Gop\Fop\vartheta_k)$,
where each $\psi_\ell$ and $\vartheta_r$ is a Boolean combination of atomic
propositions. We write $\alpha \models \varphi$ to indicate that the infinite
sequence $\alpha \in (2^{\AP})^{\omega}$ satisfies $\varphi$ under the
standard semantics.

\subsection{Weighted concurrent game structures}\label{sec:wcgs}

\begin{defi}[Arena]\label{def:arena}
  An \emph{arena} is a tuple
  \[
    \arena = \langle N, \{\Ac_i\}_{i \in N}, \{\Av_i\}_{i \in N}, \St,
      \sinit, \tr, \lab \rangle
  \]
  where $N = \{1,\dots,n\}$ is a finite non-empty set of players, $\Ac_i$ is
  the finite non-empty set of actions of player $i$, and $\St$ is a finite
  non-empty set of states with initial state $\sinit \in \St$;
  $\Av_i : \St \to 2^{\Ac_i}\setminus\{\emptyset\}$ is the \emph{protocol
  function} of player $i$, assigning to each state the non-empty set of
  actions available to $i$ there; $\lab : \St \to 2^{\AP}$ is a labelling
  function; and $\tr$ is a transition function mapping each pair consisting
  of a state $s \in \St$ and an action profile
  $\acv \in \vec{\Ac}(s) := \Av_1(s) \times \cdots \times \Av_n(s)$ to a
  successor state $\tr(s,\acv) \in \St$. For $C \subseteq N$ we write
  $\vec{\Ac}_C(s) = \prod_{i \in C}\Av_i(s)$ for the set of \emph{joint
  actions} of $C$ at $s$, and $\vec{\Ac}_{-C}(s)$ for those of the
  complement.
\end{defi}

As in standard concurrent game structures \cite{AlurHK02}, protocol
functions allow the available actions to depend on the state. In
\autoref{exa:service}, for instance, the environment cannot
play $\mathit{blocked}$ at $c$. Turn-based games can be represented by giving
every player other than the owner of a state a single available action there.

A \emph{play} from $s \in \St$ is an infinite sequence
$\out = s^{0}s^{1}s^{2}\cdots \in \St^{\omega}$ with $s^{0} = s$ such that
for every $t \geq 0$ there is $\acv^{t} \in \vec{\Ac}(s^{t})$ with
$\tr(s^{t},\acv^{t}) = s^{t+1}$. A \emph{history} is a non-empty finite
prefix of a play, and $\mathit{last}(h)$ denotes its last state. A play
$\out$ induces the sequence of labels
$\lab(\out) = \lab(s^{0})\lab(s^{1})\cdots$.

\begin{defi}[Weighted concurrent game structure]\label{def:wcgs}
  A \emph{weighted concurrent game structure} (WCGS) is a tuple
  $\game = \langle \arena, w \rangle$ where $\arena$ is an arena and
  $w : \St \to \mathbb{Z}^{d}$ is a weight function, for some $d \geq 1$. We
  write $w_j$ for the $j$-th component of $w$, $D = \{1,\dots,d\}$ for the
  set of \emph{dimensions}, and $W = \max\{|w_j(s)| : j \in D, s \in \St\}$.
  Given a play $\out$ and a dimension $j$, we write
  $\mpf_j(\out) = \mpf(w_j(\out))$ where
  $w_j(\out) = w_j(s^{0})w_j(s^{1})\cdots$. We write $\size{\game}$ for the
  size of a standard explicit representation of $\game$, including its
  transition table; in particular $\size{\vec{\Ac}(s)} \leq \size{\game}$ for
  every $s$, and the weights are encoded in binary, so $W$ may be exponential
  in $\size{\game}$.
\end{defi}

The number $d$ of dimensions need not coincide with the number $n$ of agents. A dimension may represent an individual utility or a system-level quantity such as energy consumption, latency, throughput, or social welfare. We attach weights to states, whereas mean-payoff games are often presented with edge weights \cite{ZwickP96,ChatterjeeHJ05}. A state-weighted game can be transformed straightforwardly into a polynomial-size edge-weighted game by assigning $w(s)$ to every edge leaving $s$, preserving the mean payoff of every play.

\begin{defi}[Mean-payoff constraint]\label{def:constraint}
  A \emph{mean-payoff constraint} over $D$ is generated by the grammar
  \[
    \Lambda ::= \top \;\mid\; \mpf_j \geq q \;\mid\; \Lambda \wedge \Lambda ,
  \]
  where $j \in D$ and $q \in \mathbb{Q}$. We write $\dm(\Lambda)$ for the set
  of dimensions occurring in $\Lambda$, call $\Lambda$
  \emph{one-dimensional} if $|\dm(\Lambda)| \leq 1$, and write
  $\|\Lambda\|$ for the size of its binary encoding. Satisfaction by a play
  is defined by: $\out \models \top$ always; $\out \models \mpf_j \geq q$ iff
  $\mpf_j(\out) \geq q$; and $\out \models \Lambda_1 \wedge \Lambda_2$ iff
  $\out$ satisfies both.
\end{defi}

For a threshold $q = a/b$ in lowest terms with $b > 0$, we write
$W_{q} = b\cdot W + |a|$ for the largest absolute weight occurring after the
normalisation that replaces $w_j$ by $b\cdot w_j - a$. This normalisation
turns $\mpf_j \geq q$ into $\mpf_j \geq 0$ and is used whenever we quote a
pseudo-polynomial bound. For a constraint $\Lambda$ we define
$W_{\Lambda} = \max\{W_{q} : (\mpf_j \geq q) \text{ occurs in } \Lambda\}$,
with $W_{\top} = 0$. In running-time bounds we write
$\hW_{\Lambda} = \max\{1, W_{\Lambda}\}$, so that the unconstrained case
$\Lambda = \top$ does not collapse a multiplicative bound to zero, and
similarly $\hW_{\varphi} = \max\{1, W_{\varphi}\}$.

Constraints are conjunctions of non-strict lower bounds. A robust upper
bound
\[
  \limsup_{n \to \infty}\frac{1}{n}\sum_{t<n} w_j(s^{t}) \;\leq\; q
\]
can be represented by adding the negated weight dimension and requiring its
lower mean payoff to be at least $-q$. A cost can therefore be represented
by negating the corresponding weight dimension. An upper bound on the lower mean payoff itself is not
expressible in this way.

The restriction to conjunctions reflects the intended reading of a constraint
as a collection of guarantees that must all hold on every outcome. Allowing
disjunction would give a different semantics: a coalition could satisfy
$\Lambda_1 \vee \Lambda_2$ by satisfying different disjuncts against
different counter-strategies. Consequently
$\coal{C}_{\Lambda_1 \vee \Lambda_2}\psi$ is not equivalent to
$\coal{C}_{\Lambda_1}\psi \vee \coal{C}_{\Lambda_2}\psi$, and may hold when
neither disjunct does.

We also exclude strict inequalities. For fixed-baseline deviations under
finite-memory semantics, \autoref{prop:fixed-baseline} shows that they can be
replaced by non-strict rational thresholds: a fixed finite-memory strategy
induces a finite graph whose cycles determine a positive rational margin. The
same argument does not apply to optimal values, since an optimum need not be
attained by any cycle.

\subsection{Strategies and memory}\label{sec:strategies}

\begin{defi}[Strategies]\label{def:strategies}
  A \emph{strategy} for player $i$ is a function
  $\sigma_i : \St^{+} \to \Ac_i$ with
  $\sigma_i(h) \in \Av_i(\mathit{last}(h))$ for every history $h$. Such a
  strategy is called \emph{perfect-recall}, and we write $\Sigma_i^{\PR}$ for
  the set of them.
  A strategy $\sigma_i$ is a \emph{finite-memory} strategy if it can be
  represented by a finite-state transducer
  $\sigma_i = (Q_i, q_i^{0}, \delta_i, \tau_i)$, where $Q_i$ is a finite
  non-empty set of internal states with initial state $q_i^{0}$,
  $\delta_i : Q_i \times \St \to Q_i$ is a deterministic update function, and
  $\tau_i : Q_i \times \St \to \Ac_i$ is an action function with
  $\tau_i(q,s) \in \Av_i(s)$. The transducer is read as a Mealy machine:
  along a play $s^{0}s^{1}\cdots$ the internal state evolves by
  $q_i^{t+1} = \delta_i(q_i^{t}, s^{t})$ from $q_i^{0}$, and the action taken
  at step $t$ is $\tau_i(q_i^{t}, s^{t})$. The \emph{size} of the strategy is
  $|Q_i|$, and $\Sigma_i^{\FM}$ denotes the set of finite-memory strategies.
  A \emph{memoryless} strategy is one with $|Q_i| = 1$, equivalently a
  function $\sigma_i : \St \to \Ac_i$ respecting $\Av_i$;\footnote{Under the Mealy convention, the action may depend on the current state, so one-state transducers coincide with state-based memoryless strategies.} we write
  $\Sigma_i^{\ML}$ for these. Clearly
  $\Sigma_i^{\ML} \subseteq \Sigma_i^{\FM} \subseteq \Sigma_i^{\PR}$.
\end{defi}

For $C \subseteq N$, a \emph{$C$-strategy} is a tuple
$\vec{\sigma}_C = (\sigma_i)_{i \in C}$ and a \emph{counter-strategy} is a
$(N \setminus C)$-strategy $\vec{\sigma}_{-C}$. A pair
$(\vec{\sigma}_C, \vec{\sigma}_{-C})$ is a \emph{strategy profile}
$\vec{\sigma}$. Since $\tr$ and all strategies are deterministic, a profile
$\vec{\sigma}$ and a state $s$ determine a unique play, denoted
$\out(\vec{\sigma}, s)$; we write $\out(\vec{\sigma})$ for
$\out(\vec{\sigma}, \sinit)$. We write
$\Sigma_C^{\kappa} = \prod_{i \in C}\Sigma_i^{\kappa}$ for
$\kappa \in \{\PR,\FM,\ML\}$.

Since all agents have perfect information, a $C$-strategy can equivalently be
represented by a single function
$\vec{\sigma}_C : \St^{+} \to \bigcup_{s}\vec{\Ac}_C(s)$ prescribing a
joint action after each history. We use the two representations
interchangeably and measure finite memory by the size of the corresponding
joint transducer.

\subsection{Syntax and semantics}\label{sec:syntax}

We now extend $\ATLs$ by allowing each coalition modality to carry a mean-payoff constraint.

\begin{defi}\label{def:syntax}
  Fix a finite set $\AP$ of atomic propositions, a set $N$ of agents and a
  set $D$ of dimensions. The \emph{state formulae} $\varphi$ and \emph{path
  formulae} $\psi$ of $\ATLsmp$ are given by the mutual grammar
  \begin{align*}
    \varphi &::= p \;\mid\; \neg\varphi \;\mid\; \varphi \wedge \varphi
      \;\mid\; \coal{C}_{\Lambda}\,\psi , \\
    \psi &::= \varphi \;\mid\; \neg\psi \;\mid\; \psi \wedge \psi
      \;\mid\; \Xop\psi \;\mid\; \psi \Uop \psi ,
  \end{align*}
  where $p \in \AP$, $C \subseteq N$ and $\Lambda$ is a mean-payoff
  constraint over $D$. Formulae of $\ATLsmp$ are state formulae. We use the
  usual abbreviations $\vee, \to, \Fop\psi = \top \Uop \psi$ and
  $\Gop\psi = \neg\Fop\neg\psi$, and write $\coal{C}\psi$ for
  $\coal{C}_{\top}\psi$ and $\dcoal{C}\psi$ for $\neg\coal{C}_{\top}\neg\psi$.
  We use $\dcoal{C}$ only for the qualitative case: for a non-trivial $\Lambda$, negating $\dcoal{C}_{\Lambda} \psi$ yields a condition in which either the temporal objective or the mean-payoff bounds is violated. The latter introduces a disjunction of strict upper bounds, which is not expressible by the constraint grammar of \autoref{def:constraint}.
  
  An occurrence of a strategic modality is \emph{positive} if it lies under
  an even number of negations, and a formula is \emph{positive} if all of
  its strategic modality occurrences are positive.
\end{defi}

We attach mean-payoff constraints to coalition modalities rather than admitting mean-payoff comparisons as path atoms. This keeps the quantitative component as a side condition on the outcomes of the chosen coalition strategy. Since mean payoff is prefix-independent, applying temporal operators directly to such a comparison would add no distinction on a fixed play:
\[
\Fop(\mpf_j \geq q),\qquad
\Gop(\mpf_j \geq q),\qquad\text{and}\qquad
\mpf_j \geq q
\]
would all be equivalent. Thus, \(\coal{C}_{\Lambda}\psi\) states directly that one coalition strategy must enforce both the temporal objective \(\psi\) and the mean-payoff constraint \(\Lambda\).

\begin{defi}\label{def:size}
  The \emph{size} $\size{\varphi}$ of a formula is the number of its
  subformulae plus the total size of the encodings of the constraints
  occurring in it. We write $\sub(\varphi)$ for the set of state subformulae
  of $\varphi$, and $W_{\varphi}$ for the maximum of $W_{\Lambda}$ over the
  constraints $\Lambda$ occurring in $\varphi$, with $W_{\varphi} = 0$ if
  every such constraint is $\top$.
\end{defi}

Satisfaction is parameterised by a strategy class $\kappa$ for the coalition
and $\lambda$ for its opponents. We write $\models^{\kappa,\lambda}$, and
abbreviate
\[
  \models \;=\; \models^{\PR,\PR} , \qquad
  \models_{\FM} \;=\; \models^{\FM,\FM} , \qquad
  \models_{\ML} \;=\; \models^{\ML,\PR} .
\]
In the memoryless semantics only the proponent coalition is restricted,
following the standard convention for memoryless $\ATL$ and $\ATLs$
\cite{Schobbens04,BullingDJ10,BullingJ14}. Restricting the opponents as well
can change the truth of a formula, since defeating a memoryless coalition
strategy may require memory.

\begin{defi}\label{def:semantics}
  Let $\game$ be a WCGS. State formulae are interpreted at states and path
  formulae at plays. The clauses below define $\models^{\kappa,\lambda}$;
  the superscripts are omitted where no ambiguity arises. For states,
  \begin{align*}
    \game, s &\models p
      &&\text{iff}\quad p \in \lab(s) ; \\
    \game, s &\models \neg\varphi
      &&\text{iff}\quad \game, s \not\models \varphi ; \\
    \game, s &\models \varphi_1 \wedge \varphi_2
      &&\text{iff}\quad \game, s \models \varphi_1
         \text{ and } \game, s \models \varphi_2 ;
  \end{align*}
  and, writing $\out_{s}(\vec{\sigma}_C, \vec{\sigma}_{-C})$ for
  $\out((\vec{\sigma}_C, \vec{\sigma}_{-C}), s)$,
  \begin{align*}
    \game, s \models \coal{C}_{\Lambda}\psi
    \quad&\text{iff}\quad
    \exists \vec{\sigma}_C \in \Sigma_C^{\kappa}\;
    \forall \vec{\sigma}_{-C} \in \Sigma_{-C}^{\lambda} : \\
    &\qquad
    \game, \out_{s}(\vec{\sigma}_C, \vec{\sigma}_{-C}) \models \psi
    \ \text{ and }\
    \out_{s}(\vec{\sigma}_C, \vec{\sigma}_{-C}) \models \Lambda .
  \end{align*}
  For a play $\out$,
  \begin{align*}
    \game, \out &\models \varphi
      &&\text{iff}\quad \game, \out[0] \models \varphi
         \text{, for a state formula } \varphi ; \\
    \game, \out &\models \neg\psi
      &&\text{iff}\quad \game, \out \not\models \psi ; \\
    \game, \out &\models \psi_1 \wedge \psi_2
      &&\text{iff}\quad \game, \out \models \psi_1
         \text{ and } \game, \out \models \psi_2 ; \\
    \game, \out &\models \Xop\psi
      &&\text{iff}\quad \game, \out^{\geq 1} \models \psi ; \\
    \game, \out &\models \psi_1 \Uop \psi_2
      &&\text{iff}\quad \exists k \geq 0 :
         \game, \out^{\geq k} \models \psi_2
         \text{ and } \forall\, 0 \leq t < k :
         \game, \out^{\geq t} \models \psi_1 .
  \end{align*}
  We write $\game \models \varphi$ for $\game, \sinit \models \varphi$.
\end{defi}

The temporal and quantitative conditions are evaluated on the same outcomes
of the same coalition strategy, chosen before the counter-strategy. Taking
$\Lambda = \top$ makes the second conjunct vacuous and leaves the $\ATLs$
clause of \cite{AlurHK02}, so $\ATLsmp$ is a conservative extension of
$\ATLs$.

Mean-payoff constraints are prefix-independent. Hence the truth of a
strategic formula at $s$ is independent of the history leading to $s$, and
state subformulae can be evaluated bottom-up. This would not hold for
accumulated payoff or energy constraints, whose value depends on the
preceding history \cite{BullingG22,AlechinaLNR17}.

The semantics requires the temporal and quantitative conditions to be
enforced by the same coalition strategy. The following proposition shows that
this requirement cannot in general be recovered by checking the two
abilities separately.

\begin{prop}[Non-decomposability]\label{prop:nondecomp}
  There exist a WCGS $\game$, a state $s$, a coalition $C$, a path formula
  $\psi$ and a one-dimensional constraint $\Lambda$ such that
  \[
    \game, s \models \coal{C}\psi
    \quad\text{and}\quad
    \game, s \models \coal{C}_{\Lambda}\top ,
    \quad\text{but}\quad
    \game, s \not\models \coal{C}_{\Lambda}\psi .
  \]
  Hence the combined modality is not equivalent to the conjunction of the
  corresponding qualitative and quantitative modalities.
\end{prop}

\begin{proof}
  Let $N = C = \{1\}$, $d = 1$ and $\AP = \{p\}$. The game has three states
  $s, u, v$. At $s$, agent~$1$ chooses between actions $a$ and $b$, leading
  respectively to the absorbing states $u$ and $v$; at $u$ and $v$ only one
  action is available. Let
  \[
    \lab(s) = \lab(u) = \{p\}, \qquad \lab(v) = \emptyset ,
  \]
  and $w(s) = w(u) = 0$, $w(v) = 1$. Take $\psi = \Gop p$ and
  $\Lambda = (\mpf_1 \geq 1)$.

  Playing $a$ produces the play $s\,u^{\omega}$, which satisfies $\Gop p$,
  while playing $b$ produces $s\,v^{\omega}$, whose mean payoff is $1$.
  Hence
  $\game, s \models \coal{\{1\}}\Gop p$ and
  $\game, s \models \coal{\{1\}}_{\mpf_1 \geq 1}\top$. These are the only two
  plays from $s$; the first has mean payoff $0$, and the second falsifies
  $\Gop p$ from position~$1$. No single strategy therefore satisfies both
  requirements, even in a one-player game.
\end{proof}

\section{Round-preserving sequentialisation}\label{sec:seq}

To evaluate $\coal{C}_{\Lambda}\psi$ we view $C$ as player~$1$ and its
complement as player~$2$. The coalition selects a joint action first, after
which the complement chooses its response. This standard sequentialisation
requires modification, because the inserted intermediate states affect
temporal operators containing $\Xop$. The difficulty does not arise in the
synthesis setting of \cite{BloemCHJ09,BohyBFJR13}, whose models are
turn-based.

The reduction is stated for $\LTL$ objectives over atomic propositions. For a
nested formula, state subformulae are evaluated from the inside out and
replaced by fresh propositions labelling their extensions. The following
lemma shows that this preserves satisfaction along
every play.

\begin{lem}[Substitution]\label{lem:subst}
  Let $\psi$ be a path formula with maximal state subformulae
  $\varphi_1,\dots,\varphi_k$. Let $p_1,\dots,p_k$ be fresh atomic
  propositions, extend $\lab$ by $p_i \in \lab(s)$ iff
  $s \in \ext{\varphi_i}$, and let $\psi'$ be obtained from $\psi$ by
  replacing each $\varphi_i$ by $p_i$. Then for every play $\out$ of $\game$,
  $\game, \out \models \psi$ if and only if $\lab(\out) \models \psi'$ in the
  ordinary $\LTL$ sense.
\end{lem}

\begin{proof}
  By induction on $\psi$. For $\psi = \varphi_i$ the claim is that
  $\game, \out \models \varphi_i$ iff $p_i \in \lab(\out[0])$, which holds by
  \autoref{def:semantics} and the definition of the extended labelling; the
  substitution is sound because the truth of a state subformula at a state
  does not depend on the history leading to it.
  The cases of $\neg, \wedge, \Xop, \Uop$ are immediate, since the semantics
  of these operators coincides with their $\LTL$ semantics on the label
  sequence and commutes with taking suffixes.
\end{proof}

\subsection{Coalition sequentialisation}

\begin{defi}\label{def:seq}
  Let $\game = \langle \arena, w\rangle$ be a WCGS and let $C \subseteq N$.
  The \emph{sequentialisation of $\game$ with respect to $C$} is the
  turn-based two-player weighted arena
  $\game^{C} = \langle V_1, V_2, E, \hat{w}, \hat{\lab}\rangle$ where
  \begin{align*}
    V_1 &= \St , &
    V_2 &= \{(s,\alpha) : s \in \St,\ \alpha \in \vec{\Ac}_C(s)\} ,
  \end{align*}
  and $E$ consists of the edges
  $(s, (s,\alpha))$ for every $s \in \St$ and $\alpha \in \vec{\Ac}_C(s)$,
  together with the edges $((s,\alpha), s')$ such that
  $\tr(s, (\alpha,\beta)) = s'$ for some $\beta \in \vec{\Ac}_{-C}(s)$.
  The weight and labelling functions are lifted by
  $\hat{w}(s) = \hat{w}(s,\alpha) = w(s)$ and
  $\hat{\lab}(s) = \hat{\lab}(s,\alpha) = \lab(s)$. Player~$1$ owns $V_1$ and
  player~$2$ owns $V_2$. When $C = N$ the set $\vec{\Ac}_{-C}(s)$ is a
  singleton and player~$2$ has no choice.
\end{defi}

The construction is polynomial: $|V_1| + |V_2| \leq \size{\game}^{2}$, and
$E$ is read off the transition table. Plays of $\game^{C}$ from a vertex in
$V_1$ alternate strictly between $V_1$ and $V_2$, so they have the shape
$\rho = s^{0}(s^{0},\alpha^{0})s^{1}(s^{1},\alpha^{1})s^{2}\cdots$.

\begin{defi}\label{def:proj}
  The \emph{projection} of a play $\rho = v^{0}v^{1}v^{2}\cdots$ of
  $\game^{C}$ starting in $V_1$ is the sequence
  $\proj(\rho) = v^{0}v^{2}v^{4}\cdots \in \St^{\omega}$ obtained by deleting
  the $V_2$-vertices.
\end{defi}

\begin{prop}\label{prop:projection}
  Let $s \in \St$. Projection maps every play of $\game^{C}$ from $s$ to a
  play of $\game$ from $s$, and every play of $\game$ from $s$ is the
  projection of some play of $\game^{C}$ from $s$. Moreover, projection
  preserves every mean-payoff dimension:
  \[
    \mpf(\hat{w}_j(\rho)) = \mpf_j(\proj(\rho))
  \]
  for every play $\rho$ of $\game^{C}$ and every dimension $j$. Projection
  need not be injective; however, after fixing a player~$1$ strategy, every
  play of $\game$ has at most one preimage consistent with that strategy.
\end{prop}

\begin{proof}
  Legality of the projection follows from the definition of $E$: consecutive
  edges $s^{t} \to (s^{t},\alpha^{t}) \to s^{t+1}$ witness some
  $\beta^{t} \in \vec{\Ac}_{-C}(s^{t})$ with
  $\tr(s^{t},(\alpha^{t},\beta^{t})) = s^{t+1}$. A preimage is obtained by
  choosing, for each original transition, a witnessing action profile and
  retaining its $C$-component. Each original state weight is duplicated
  exactly twice in the sequentialised play, so the two sequences of running
  averages have the same limit inferior; the calculation is in
  \appref{app:seq}. Uniqueness under a fixed player~$1$ strategy follows by
  induction, since the intermediate joint action is then determined by the
  preceding history.
\end{proof}

\autoref{prop:projection} shows that the usual construction suffices for
the quantitative fragment. It does not preserve arbitrary temporal
objectives. Each original state is duplicated by an inserted vertex carrying
the same label, so the sequentialised play is a stuttered copy of the
original. Evaluating $\Xop p$ at $s^{0}$ in $\game$ tests $p$ at $s^{1}$,
whereas evaluating it on the sequentialised play tests $p$ at the
intermediate vertex $(s^{0},\alpha^{0})$, which is labelled like $s^{0}$.
Similarly, $\Xop\Xop p$ tests $p$ at $s^{1}$ rather than at $s^{2}$. $\LTL$
without $\Xop$ is stutter-invariant, so those formulae are unaffected; the
full logic requires the temporal specification to advance once per concurrent
\emph{round} rather than once per product vertex.

\subsection{Round-preserving product}

We therefore evaluate the automaton on the projected play, realising the
projection inside an automaton product.
Recall that every $\LTL$ formula $\psi$ over $\AP$ can be translated into a
deterministic parity automaton (DPA)
$\dpa_{\psi} = (Q, 2^{\AP}, q^{0}, \delta, \pri)$ with
$L(\dpa_{\psi}) = \{\xi \in (2^{\AP})^{\omega} : \xi \models \psi\}$,
$|Q| \in 2^{2^{O(\size{\psi})}}$ and $|\pri(Q)| \in 2^{O(\size{\psi})}$
\cite{Safra88,Piterman07,EsparzaKRS22}. A run is accepting if the least
priority occurring infinitely often is even.

\begin{defi}[Round-preserving product]\label{def:rpp}
  Let $\game$ be a WCGS, $C \subseteq N$, $\psi$ an $\LTL$ formula over
  $\AP$, and $\dpa_{\psi} = (Q, 2^{\AP}, q^{0}, \delta, \pri)$ a DPA for
  $\psi$. The \emph{round-preserving product}
  $\game^{C} \otimes \dpa_{\psi}$ is the two-player weighted parity arena
  $\langle U_1, U_2, F, \bar{w}, \bar{\pri}\rangle$ with
  $U_1 = \St \times Q$ and $U_2 = V_2 \times Q$, and with transitions
  \begin{align*}
    \bigl((s,q),\, ((s,\alpha),q)\bigr) &\in F
      &&\text{for all } \alpha \in \vec{\Ac}_C(s) , \\
    \bigl(((s,\alpha),q),\, (s', \delta(q, \lab(s)))\bigr) &\in F
      &&\text{whenever } ((s,\alpha), s') \in E ,
  \end{align*}
  where $\bar{w}(u) = w(s)$ for any vertex $u$ whose $\St$-component is $s$,
  and $\bar{\pri}(s,q) = \bar{\pri}((s,\alpha),q) = \pri(q)$. Player~$1$ owns
  $U_1$ and player~$2$ owns $U_2$; the initial vertex for a state $s$ is
  $(s,q^{0})$.
\end{defi}

The automaton is updated once per concurrent round, on the label of the
original state in which the round began. Its priority is copied to both
product vertices of that round. The parity condition on a product play
therefore agrees with acceptance of the automaton run on the projected
play.

\begin{lem}[Projection lemma]\label{lem:projection}
  Let $\bar{\rho}$ be a play of $\game^{C} \otimes \dpa_{\psi}$ from
  $(s, q^{0})$, let $\rho$ be its projection to $\game^{C}$ and let
  $\out = \proj(\rho)$. Then:
  \begin{enumerate}[\em(1)]
  \item the sequence of automaton components of $\bar{\rho}$, read at the
    $U_1$-vertices, is exactly the run of $\dpa_{\psi}$ on $\lab(\out)$;
  \item $\bar{\rho}$ satisfies the parity condition given by $\bar{\pri}$ if
    and only if $\out \models \psi$;
  \item $\mpf(\bar{w}_j(\bar{\rho})) = \mpf_j(\out)$ for every dimension $j$.
  \end{enumerate}
\end{lem}

\begin{proof}
  The automaton component is updated once per round, on the player~$2$ move,
  and reads $\lab(s^{t})$, so the $U_1$-vertices of $\bar{\rho}$ carry
  exactly the run of $\dpa_{\psi}$ on $\lab(\out)$, giving (1). Priorities
  are duplicated by the passage from $U_1$- to $U_2$-vertices but none is
  introduced or removed, so the least priority occurring infinitely often
  agrees with that of the automaton run; since $\dpa_{\psi}$ is deterministic
  and recognises the models of $\psi$, this gives (2). Weights depend only on
  the $\St$-component, so (3) follows from \autoref{prop:projection}.
  Details are in \appref{app:seq}.
\end{proof}

\subsection{Strategy correspondence}

Let $\mathcal{A}$ be a turn-based two-player arena whose vertices carry a
priority function $\bar{\pri}$ and a weight function
$\bar{w} : V \to \mathbb{Z}^{d}$. We write $\mathrm{Par}$ for the set of
plays of $\mathcal{A}$ whose least priority occurring infinitely often is
even, and, for a mean-payoff constraint $\Lambda$ over the $d$ dimensions,
$\mathrm{Par}\wedge\Lambda$ for the set of plays satisfying both. A
\emph{mean-payoff parity objective} for player~$1$ is an objective of this
form; the product $\game^{C}\otimes\dpa_{\psi}$ of \autoref{def:rpp} is
the instance used here.

\begin{thm}[Strategy correspondence]\label{thm:seq}
  Let $\game$ be a WCGS, $s \in \St$, $C \subseteq N$, $\Lambda$ a
  mean-payoff constraint and $\psi$ an $\LTL$ formula over $\AP$. Then,
  under perfect-recall semantics,
  \[
    \game, s \models \coal{C}_{\Lambda}\psi
    \quad\text{iff}\quad
    \text{player~$1$ wins } \game^{C} \otimes \dpa_{\psi} \text{ from }
    (s, q^{0}) \text{ for } \mathrm{Par}\wedge\Lambda .
  \]
\end{thm}

\begin{proof}
  Write $\Pi = \game^{C}\otimes\dpa_{\psi}$. By \autoref{lem:projection}, a
  play of $\Pi$ satisfies $\mathrm{Par}\wedge\Lambda$ iff its projection
  satisfies $\psi\wedge\Lambda$. It therefore remains to relate the
  strategies in the two games.

  \emph{From $\game$ to $\Pi$.} Let $\vec{\sigma}_C$ witness
  $\game, s \models \coal{C}_{\Lambda}\psi$. At each product history,
  player~$1$ plays the joint action prescribed by $\vec{\sigma}_C$ on its
  projection. Any player~$2$ strategy then induces a counter-strategy in
  $\game$: on each prefix of the projected play, choose an action profile of
  the complement witnessing the corresponding transition, and define the
  strategy arbitrarily elsewhere. The projected play is therefore an outcome
  of $\vec{\sigma}_C$ and satisfies $\psi\wedge\Lambda$.

  \emph{From $\Pi$ to $\game$.} Let $\sigma_1$ be winning in $\Pi$ from
  $(s,q^{0})$. By \autoref{prop:projection}, each history of $\game$ has
  at most one preimage consistent with $\sigma_1$. On such a history, let
  $\vec{\sigma}_C$ play the joint action prescribed by $\sigma_1$, and define
  it arbitrarily otherwise. Every counter-strategy in $\game$ induces a
  player~$2$ strategy in $\Pi$ whose outcome projects to the corresponding
  outcome of $\vec{\sigma}_C$. Since $\sigma_1$ is winning,
  \autoref{lem:projection} implies that this projected play satisfies
  $\psi\wedge\Lambda$.

  The constructions are given in full in \appref{app:seq}.
\end{proof}

\begin{rem}[Quantifier order]\label{rem:quantifier}
  Although player~$2$ observes the coalition action $\alpha^{t}$ in
  $\game^{C}$, this does not add information unavailable to the complement in
  $\game$. The counter-strategy is quantified after $\vec{\sigma}_C$ and can
  therefore compute $\alpha^{t} = \vec{\sigma}_C(h^{t})$ from the history.
  The argument relies on deterministic strategies, deterministic transitions
  and perfect information; it does not extend directly to randomised
  strategies.
\end{rem}

Finite-memory strategies require some additional care. A strategy in the
product updates its memory at both vertices of each sequentialised round,
whereas a coalition strategy in $\game$ updates only once per original state.
We must also show that, against a fixed finite-memory coalition strategy,
finite-memory counter-strategies suffice.

\begin{thm}[Finite-memory correspondence]\label{thm:fm-correspondence}
  Let $\Lambda$ be a mean-payoff constraint, let $\psi$ be an $\LTL$ formula,
  and let $\Pi = \game^{C}\otimes\dpa_{\psi}$, where $\dpa_{\psi}$ has state
  space $Q$. Then:
  \begin{enumerate}[\em(1)]
  \item if a finite-memory $C$-strategy of size $m$ enforces
    $\psi\wedge\Lambda$ from $s$ against every counter-strategy, then
    player~$1$ has a finite-memory strategy of size $m$ winning
    $\mathrm{Par}\wedge\Lambda$ in $\Pi$ from $(s,q^{0})$;
  \item if player~$1$ has a finite-memory strategy of size $m$ winning
    $\mathrm{Par}\wedge\Lambda$ in $\Pi$ from $(s,q^{0})$, then $C$ has a
    finite-memory strategy of size at most $m\cdot|Q|$ enforcing
    $\psi\wedge\Lambda$ from $s$ against every counter-strategy;
  \item for every finite-memory $C$-strategy, if some counter-strategy
    produces an outcome from $s$ violating $\psi\wedge\Lambda$, then some
    finite-memory counter-strategy does.
  \end{enumerate}
  Consequently $\models^{\FM,\FM}$ and $\models^{\FM,\PR}$ agree on every
  $\ATLsmp$ formula, and $\game, s \models_{\FM}\coal{C}_{\Lambda}\psi$ iff
  player~$1$ has a finite-memory strategy winning
  $\mathrm{Par}\wedge\Lambda$ in $\Pi$ from $(s,q^{0})$.
\end{thm}

\begin{proof}
  Items~(1) and~(2) refine the strategy translations of \autoref{thm:seq}.

  For~(1), let $\vec{\sigma}_C$ be represented by a transducer of size $m$.
  At a $U_1$-vertex whose state component is $s'$, the induced player~$1$
  strategy plays the joint action prescribed by $\vec{\sigma}_C$ at $s'$ and
  performs the corresponding memory update. At the following $U_2$-vertex it
  leaves its memory unchanged. Its memory at each $U_1$-vertex therefore
  agrees with the memory of $\vec{\sigma}_C$ on the projected history, and no
  additional memory is required.

  For~(2), let $\sigma_1$ be represented by a transducer with memory set $M$.
  The translated coalition strategy stores a pair $(r,q) \in M\times Q$,
  where $r$ is the memory state of $\sigma_1$ and $q$ is the current
  automaton state. At an original state $s'$ it plays the joint action
  $\alpha$ prescribed by $\sigma_1$ at $(s',q)$. It then simulates, in one
  update, the two product updates at $(s',q)$ and at $((s',\alpha),q)$, and
  stores the resulting memory state together with $\delta(q,\lab(s'))$. This
  yields a transducer of size at most $m\cdot|Q|$. In both directions,
  correctness follows from \autoref{lem:projection}.

  For~(3), fix a finite-memory $C$-strategy and the induced player~$1$
  strategy in $\Pi$. Fixing this strategy yields a finite graph $H$. If some
  counter-strategy produces an outcome violating $\psi\wedge\Lambda$, its
  corresponding path in $H$ satisfies
  $\neg\mathrm{Par} \vee \neg\Lambda$. A parity violation in a finite graph
  has an ultimately periodic witness. Otherwise some conjunct
  $\mpf_j \geq q_j$ of $\Lambda$ is violated. If every reachable cycle of $H$
  had mean at least $q_j$ in dimension $j$, the usual decomposition of a
  finite path into cycles and a residual simple path would imply that every
  infinite path has $\liminf$ mean at least $q_j$. Hence $H$ contains a
  reachable cycle of mean below $q_j$, and a lasso reaching and repeating
  that cycle violates $\Lambda$. In either case a violating lasso is realised
  by a finite-memory player~$2$ strategy. Tracking this lasso together with
  the memory of the fixed coalition strategy gives a finite-memory
  counter-strategy in $\game$.

  The agreement of $\models^{\FM,\FM}$ and $\models^{\FM,\PR}$ for arbitrary
  formulae follows by structural induction. For a strategic formula, the
  induction hypothesis gives the same extensions to its maximal state
  subformulae under both semantics; \autoref{lem:subst} therefore gives the
  same substituted $\LTL$ objective, and item~(3) applies. The final
  equivalence follows from items~(1) and~(2).

  Full details are given in \appref{app:seq}.
\end{proof}

\section{Model checking}\label{sec:mc}

The model-checking problem is to decide, given a WCGS $\game$, a state $s$,
and an $\ATLsmp$ formula $\varphi$, whether $\game, s \models \varphi$.

\subsection{Model-checking procedure}

The model-checking procedure follows the standard bottom-up evaluation of
$\ATLs$ \cite{AlurHK02}. Strategic subformulae are evaluated by solving the
games obtained from the construction of \autoref{sec:seq}. We first treat
perfect-recall and finite-memory semantics, which use the same product
construction but different game-solving characterisations; memoryless
semantics requires a separate argument and is considered in
\autoref{sec:mc-ml}.

Let $\varphi$ be an $\ATLsmp$ formula. We compute the extension of each state
subformula in increasing order of nesting depth, handling Boolean
subformulae by complement and intersection. For a strategic subformula
$\chi = \coal{C}_{\Lambda}\psi$:
\begin{enumerate}[(1)]
\item replace each maximal state subformula of $\psi$ by a fresh atomic
  proposition labelling its previously computed extension;
\item translate the resulting $\LTL$ formula $\psi'$ into a deterministic
  parity automaton $\dpa_{\psi'}$;
\item construct the round-preserving product
  $\Pi = \game^{C}\otimes\dpa_{\psi'}$ of \autoref{def:rpp};
\item compute the winning region for $\mathrm{Par}\wedge\Lambda$, requiring
  a finite-memory winning strategy under finite-memory semantics;
\item set the extension of $\chi$ to the states $s$ such that $(s,q^{0})$ is
  winning for player~$1$.
\end{enumerate}
After all state subformulae have been evaluated, the formula holds at $s$ iff
$s \in \ext{\varphi}$. 

\paragraph{Correctness.}
The proof is by induction on the nesting depth of strategic modalities. The
atomic and Boolean cases are immediate. Consider
$\chi = \coal{C}_{\Lambda}\psi$, and suppose that the extensions of the
maximal state subformulae of $\psi$ have already been computed. By
\autoref{lem:subst}, a play $\out$ satisfies $\psi$ iff its label sequence
satisfies the substituted $\LTL$ formula $\psi'$. By
\autoref{lem:projection}, the plays of $\Pi = \game^{C}\otimes\dpa_{\psi'}$
satisfying $\mathrm{Par}\wedge\Lambda$ are exactly those whose projections
satisfy $\psi\wedge\Lambda$. The strategy correspondence of
\autoref{thm:seq} therefore identifies the winning region of $\Pi$ with the
extension of $\chi$. Under finite-memory semantics, the same conclusion
follows from \autoref{thm:fm-correspondence}.

\subsection{One-dimensional game solving}\label{sec:mc-1d}

For one-dimensional constraints, each strategic subformula reduces to a
one-dimensional mean-payoff parity game. We recall the bounds needed for the
complexity analysis. One-dimensional mean-payoff games are memorylessly
determined, and their threshold problem lies in $\NP\cap\coNP$
\cite{EhrenfeuchtM79,ZwickP96}. For mean-payoff parity games, deciding
whether player~$1$ wins $\mathrm{Par}\wedge(\mpf \geq q)$ is polynomially
equivalent to solving an energy parity game and also lies in $\NP\cap\coNP$
\cite[Theorem~4 and Corollary~1]{ChatterjeeD12}. Player~$2$ has memoryless
optimal strategies, whereas player~$1$ may require infinite memory
\cite{ChatterjeeHJ05}. The winning region can be computed in time
\begin{equation}\label{eq:mpp-time}
  O\bigl(|E|\cdot c\cdot|V|^{c+2}\cdot \hW_{q}\cdot(|V|+1)\bigr) ,
\end{equation}
where $c$ is the number of priorities \cite[Corollary~2]{ChatterjeeD12} and
$\hW_{q} = \max\{1,W_q\}$ for the normalised weight bound $W_q$ of
\autoref{sec:model}. The bound is
therefore pseudo-polynomial in the encoding of the threshold as well as of
the weights.

A one-dimensional constraint may contain several bounds on the same
dimension. If $\Lambda \neq \top$, it is equivalent to $\mpf_j \geq q^{*}$,
where $q^{*}$ is the largest threshold occurring in $\Lambda$. We assume this
normal form throughout this section.

We use this pseudo-polynomial bound below.\footnote{A pseudo-quasi-polynomial
algorithm is given in \cite{DaviaudJL18}, but the bound in \autoref{lem:fm-energy} suffice for our purposes.} The finite-memory case requires a
different characterisation.

\begin{lem}\label{lem:fm-energy}
  Let $\game^{\bullet}$ be a weighted parity game with one weight dimension,
  $|V|$ vertices, $c$ priorities and largest absolute weight $W$. Player~$1$
  has a \emph{finite-memory} strategy winning $\mathrm{Par}\wedge(\mpf\geq0)$
  from $v$ if and only if there exists a finite initial credit from which
  player~$1$ wins the energy parity game on the same arena from $v$.
  Consequently the
  finite-memory threshold problem is in $\NP\cap\coNP$, and whenever a
  finite-memory winning strategy exists there is one of memory size
  $O(|V|\cdot c\cdot W)$.
\end{lem}

\begin{proof}
  Suppose first that player~$1$ wins the energy parity game with initial
  credit $c_0$. Along every consistent play each partial sum is at least
  $-c_0$, and hence the $\liminf$ mean payoff is at least $0$; the same
  strategy satisfies the parity condition. Moreover, whenever the energy
  parity game is winning, player~$1$ has a strategy of memory size
  $O(|V|\cdot c\cdot W)$ \cite[Theorem~1]{ChatterjeeD12}.

  Conversely, fix a finite-memory strategy winning
  $\mathrm{Par}\wedge(\mpf\geq0)$ and consider the finite graph obtained by
  fixing that strategy. Every reachable cycle has non-negative total weight:
  otherwise player~$2$ could reach and repeat a negative cycle, producing a
  consistent play with negative mean payoff. The running weight along every
  path is therefore bounded below, since each finite path decomposes into
  non-negative cycles and a residual simple path. The strategy thus wins the
  energy parity game from some finite initial credit.

  The $\NP\cap\coNP$ bound and the stated memory bound now follow from the
  corresponding results for energy parity games
  \cite[Theorems~1 and~2]{ChatterjeeD12}. Details are in
  \appref{app:energy}.
\end{proof}

The equivalence in \autoref{lem:fm-energy} holds at threshold $0$, without
shifting the weights. For unrestricted mean-payoff parity games, the positive
shift in \cite[Theorem~4]{ChatterjeeD12} is needed because an exact threshold
may be achievable only with infinite memory. In the absence of a parity
condition, the zero-threshold equivalence extends to multiple dimensions
\cite[Theorem~3]{ChatterjeeDHR10}.

\begin{thm}\label{thm:main}
  Model checking $\ATLsmp$ with one-dimensional constraints is
  $\TwoExp$-complete. This holds under both perfect-recall and
  finite-memory semantics.
\end{thm}

\begin{proof}
  \emph{Upper bound.} Put $m = \size{\varphi}$ and
  $n_{\game} = \size{\game}$. The procedure evaluates at most $m$ strategic
  subformulae, and the labelling is extended by at most $m$ fresh
  propositions, so the arena never grows.

  Fix one strategic subformula $\coal{C}_{\Lambda}\psi$ and let $\psi'$ be as
  by substitution, with $\size{\psi'} \leq m$. Determinisation produces
  a DPA with
  $|Q| \leq 2^{2^{O(m)}}$ states and $c \leq 2^{O(m)}$ priorities
  \cite{Safra88,Piterman07,EsparzaKRS22}. The round-preserving product has
  $|U_1| + |U_2| \leq n_{\game}^{2}\cdot|Q|$ vertices. By
\eqref{eq:mpp-time}, solving it takes time
  \[
    O\Bigl(\bigl(n_{\game}^{2}\,2^{2^{O(m)}}\bigr)^{c+O(1)}
      \cdot \hW_{\varphi}\Bigr)
    \;=\; 2^{\,2^{O(m)} \,+\, 2^{O(m)}\log n_{\game}}
      \cdot \mathrm{poly}(\hW_{\varphi}) ,
  \]
  using $c \leq 2^{O(m)}$ and
  $\bigl(2^{2^{O(m)}}\bigr)^{2^{O(m)}} = 2^{2^{O(m)}}$. The product arena is
  already doubly exponential, but it is raised only to a singly exponential
  number of priorities, and
  $\bigl(2^{2^{am}}\bigr)^{2^{bm}} = 2^{2^{(a+b)m}}$. Although
  $n_{\game}$ is raised to the number $c$ of priorities, the resulting bound
  is doubly exponential in the combined input size $n$. Indeed
  $c \leq 2^{O(m)}$, $\log n_{\game} \leq n$ and
  $\log \hW_{\varphi} \leq O(n)$. Summing over at most $m$ strategic
  subformulae preserves the bound.

  \emph{Finite memory.} By \autoref{thm:fm-correspondence} the truth of $\coal{C}_{\Lambda}\psi$ at $s$
  under $\models_{\FM}$ is equivalent to player~$1$ having a finite-memory
  winning strategy in $\Pi$. By \autoref{lem:fm-energy} that question is
  decided by solving an energy parity game on $\Pi$. The arena and priority
  function are unchanged, and that problem has the same pseudo-polynomial
  dependence on $W_{q}$, so the same doubly exponential bound applies.

  \emph{Lower bound.} Set every constraint to $\top$ and every weight to
  $0$. By conservativity this gives a linear reduction from $\ATLs$ model
  checking over concurrent game structures. For such formulae, perfect-recall
  and finite-memory semantics coincide, since an $\omega$-regular objective
  that can be enforced over a finite game has a finite-memory witness
  \cite{AlurHK02,BullingJ14}. As $\ATLs$ model checking is $\TwoExp$-hard
  \cite{AlurHK02}, the $\TwoExp$ lower bound holds under both semantics.
\end{proof}

\subsection{Memoryless semantics}\label{sec:mc-ml}

The product construction does not preserve memorylessness. A positional
strategy in $\game^{C}\otimes\dpa_{\psi}$ may choose different actions at
$(s,q)$ and $(s,q')$, whereas a memoryless strategy in $\game$ must choose
the same action whenever the current game state is $s$. Translating such a
product strategy therefore generally introduces finite memory. We instead
give a direct model-checking procedure.

\begin{thm}\label{thm:memoryless}
  Model checking $\ATLsmp$ under memoryless semantics is $\PSPACE$-complete,
  already when every constraint is $\top$. The $\PSPACE$ upper bound holds
  for arbitrary conjunctive mean-payoff constraints.
\end{thm}

\begin{proof}
  \emph{Upper bound.} Evaluate state subformulae bottom-up. For a strategic
  subformula $\coal{C}_{\Lambda}\psi$, replace the maximal state subformulae
  of $\psi$ as in \autoref{lem:subst}, obtaining an $\LTL$ formula $\psi'$.
  Guess a memoryless joint strategy $f_C(s) \in \vec{\Ac}_C(s)$ for each
  state $s$; this strategy has a representation polynomial in
  $\size{\game}$.

  Let $\game[f_C]$ be the graph with an edge $s \to s'$ whenever
  $\tr(s,(f_C(s),\beta)) = s'$ for some action $\beta$ of the complementary
  coalition. Its paths from $s$ are exactly the outcomes of $f_C$: every
  counter-strategy induces such a path, and every such path is realised by
  choosing a witnessing $\beta$ after each prefix, which is legitimate by the
  quantifier-order argument of \autoref{rem:quantifier}.

  It remains to check that every path from $s$ satisfies
  $\psi' \wedge \Lambda$. Universal $\LTL$ model checking is in $\PSPACE$
  \cite{SistlaC85}. For each conjunct $\mpf_j \geq q_j$, normalise the
  threshold to $0$ by replacing $w_j$ with $b\cdot w_j - a$ for $q_j = a/b$.
  Every path then has $\liminf$ mean at least $0$ in dimension $j$ iff every
  cycle reachable from $s$ has non-negative total weight in that dimension:
  a negative such cycle traversed forever is a violating lasso, and
  conversely the running sum along any prefix is bounded below by the
  contribution of a simple path. For each dimension $j$, the existence of a
  reachable cycle of negative total normalised weight can be checked in
  polynomial time---for example by Bellman--Ford negative-cycle detection,
  equivalently by Karp's minimum-cycle-mean algorithm, after restricting
  $\game[f_C]$ to the vertices reachable from $s$---so all quantitative
  conjuncts are verified in time $O(d\cdot|V|\cdot|E|)$ in the size of
  $\game[f_C]$.
  Guessing $f_C$ therefore gives an $\mathrm{NPSPACE}$ procedure, and
  $\mathrm{NPSPACE} = \PSPACE$. Repeating the argument bottom-up for all
  state subformulae preserves the bound.

  \emph{Lower bound.} Take $C = \emptyset$ and $\Lambda = \top$. The empty
  coalition has a unique strategy, and its outcomes are exactly the plays of
  the arena, each obtained by having every agent follow the play along its
  prefixes. Hence $\coal{\emptyset}\psi$ holds at $s$ iff every play from $s$
  satisfies $\psi$, which is universal $\LTL$ model checking and
  $\PSPACE$-hard \cite{SistlaC85}. The same bound is known for memoryless
  $\ATLs$ model checking directly \cite[Theorem~13]{BullingDJ10}.
\end{proof}

\section{Refined bounds for one-dimensional fragments}\label{sec:fragments}

Throughout this section, constraints are one-dimensional. The $\TwoExp$ upper
bound of \autoref{thm:main} is driven by the
determinisation of the temporal objective. We now consider fragments in which
the temporal component admits a smaller monitor or a direct game
construction. \autoref{tab:summary} summarises the resulting bounds.

\label{sec:frag-defs}%
We distinguish two independent restrictions: one on the form of the path
objective, and one on the Boolean and nesting structure of strategic
modalities.

\begin{defi}\label{def:fragments}
  According to the form of the path objective, we distinguish:
  \begin{itemize}
  \item $\ATLsmp$, the full logic of \autoref{def:syntax};
  \item $\ATLmp$, the fragment in which every temporal operator occurs
    immediately within a strategic modality; equivalently, strategic
    modalities have the form
    \[
      \coal{C}_{\Lambda}\Xop\varphi , \qquad
      \coal{C}_{\Lambda}\Gop\varphi , \qquad\text{or}\qquad
      \coal{C}_{\Lambda}(\varphi_1 \Uop \varphi_2) ,
    \]
    where $\varphi, \varphi_1, \varphi_2$ are state formulae;
  \item $\ATLsmpGR$, in which every strategic modality is applied to a $\GR$
    formula over state formulae;
  \item the \emph{quantitative fragment}, consisting of the Boolean
    combinations of atomic propositions and \emph{quantitative strategic
    atoms} $\coal{C}_{\Lambda}\top$.
  \end{itemize}
  Orthogonally, we call a formula a \emph{strategic atom} if it is of the form
  $\coal{C}_{\Lambda}\psi$ with $\psi$ free of strategic modalities;
  \emph{flat} if it is a Boolean combination of atomic propositions and
  strategic atoms; and
  \emph{nested} otherwise.
\end{defi}

Since the quantitative fragment fixes every path formula to $\top$, it is
flat and has no nested case.

We write $\MPG$ for the one-dimensional mean-payoff threshold problem on
two-player turn-based games and $\MPP$ for its mean-payoff parity
counterpart. Both lie in $\NP\cap\coNP$
\cite{ZwickP96,ChatterjeeHJ05,ChatterjeeD12}, and neither is known to be in
$\PTIME$. The notation $\PTIME^{\MPG}$ and $\PTIME^{\MPP}$ denotes
polynomial time with the corresponding oracle.

Since an $\NP\cap\coNP$ atom has polynomial certificates for both truth and
falsity, a polynomial-size Boolean combination of such atoms remains in
$\NP\cap\coNP$: guess the truth value of each atom together with the
corresponding certificate, and then evaluate the Boolean combination. We use
this closure property for flat formulae throughout. The argument does not
apply directly to multi-dimensional atoms, which are treated separately in
\autoref{sec:multi}.

\begin{table}[!ht]
  \centering
  \caption{Refined complexity bounds for one-dimensional fragments. Entries
    are for the fragments of \autoref{def:fragments} with one-dimensional
    constraints, and hold under both perfect-recall and finite-memory
    semantics. $\MPG$ and $\MPP$ denote the threshold problems for
    one-dimensional mean-payoff games and mean-payoff parity games
    respectively.}
  \label{tab:summary}
  \small
  \setlength{\tabcolsep}{4.5pt}
  \begin{tabular}{lccc}
    \toprule
    Fragment & Strategic atom & Flat formula & Nested formula \\
    \midrule
    Quantitative, $\psi = \top$
      & $\NP\cap\coNP$, $\MPG$-equivalent
      & $\NP\cap\coNP$
      & --- \\
    $\ATLmp$
      & $\NP\cap\coNP$, $\MPG$-hard
      & $\NP\cap\coNP$
      & $\PTIME^{\MPP}$ \\
    $\ATLsmpGR$
      & $\NP\cap\coNP$, $\MPG$-hard
      & $\NP\cap\coNP$
      & $\PTIME^{\MPP}$ \\
    $\ATLsmp$
      & $\TwoExp$-complete
      & $\TwoExp$-complete
      & $\TwoExp$-complete \\
    \bottomrule
  \end{tabular}
\end{table}

\subsection{Direct constructions}\label{sec:frag-1d}

For the quantitative fragment no temporal monitor is needed: taking
$\psi = \top$ collapses the product to $\game^{C}$, so an atom is a
one-dimensional mean-payoff threshold query on the sequentialised arena.

In $\ATLmp$ each temporal operator occurs immediately within a strategic
modality, so the doubly exponential $\LTL$ translation is unnecessary: each
of the three operators admits a direct construction.

\begin{lem}\label{lem:atl-constructions}
  Let $\varphi, \varphi_1, \varphi_2$ be state formulae whose extensions are
  already computed and let $\Lambda$ be one-dimensional. Then
  $\coal{C}_{\Lambda}\chi$ can be decided at all states simultaneously by
  solving a single mean-payoff parity game with at most two priorities on an
  arena of size $O(\size{\game}^{2})$, for
  $\chi \in \{\Xop\varphi,\ \Gop\varphi,\ \varphi_1 \Uop \varphi_2\}$.
\end{lem}

\begin{proof}
  For $\Xop\varphi$, compute the winning region $R$ of the mean-payoff game
  on $\game^{C}$; a state $s$ is winning iff the coalition has an action
  $\alpha$ with $(s,\alpha) \in R$ and every successor of $(s,\alpha)$ has
  state component in $\ext{\varphi}$. For $\Gop\varphi$, compute the maximal
  region in which player~$1$ can keep the state component inside
  $\ext{\varphi}$, a safety region obtained by attractor computation; every
  strategy enforcing $\Gop\varphi$ remains in this region, so it remains only
  to enforce $\Lambda$ there. For $\varphi_1 \Uop \varphi_2$, use two copies
  recording whether $\varphi_2$ has been reached, with the initial copy
  recording whether $\varphi_2$ already holds at $s$, together with a losing
  sink for leaving $\varphi_1$ beforehand. The last construction uses two
  priorities and the first two use at most two, each on an arena of size
  $O(\size{\game}^{2})$. The constructions are given in \appref{app:frag}.
\end{proof}

$\GR$ is a widely used fragment for reactive synthesis
\cite{BloemJPPS12,KressGazitFP09}. For such objectives the general $\LTL$
determinisation can be replaced by a polynomial-size deterministic parity
monitor.

\begin{lem}[$\GR$ monitor]\label{lem:gr1-dpa}
  Let
  \[
    \theta \;=\;
    \Bigl(\bigwedge_{\ell=1}^{m}\Gop\Fop\psi_{\ell}\Bigr)
    \;\to\;
    \Bigl(\bigwedge_{r=1}^{k}\Gop\Fop\vartheta_{r}\Bigr) ,
  \]
  where each $\psi_{\ell}$ and $\vartheta_{r}$ is a Boolean formula over
  $\AP$ and $m, k \geq 1$. There is a symbolic deterministic parity automaton
  recognising the models of $\theta$, with $O(mk)$ states and three
  priorities. Its transitions are represented by Boolean guards over $\AP$
  and can be constructed in time polynomial in $\size{\theta}$. The cases
  $m = 0$ and $k = 0$ are immediate.
\end{lem}

\begin{proof}
  Use a cyclic counter over $\{1,\dots,m\}$ to record the next assumption
  $\psi_{\ell}$ to be witnessed. Whenever the current letter satisfies that
  assumption, the counter advances; let $a$ denote the event that the counter
  completes a cycle. Then $\bigwedge_{\ell}\Gop\Fop\psi_{\ell}$ holds iff
  $\Gop\Fop a$. Similarly, a cyclic counter over the guarantees produces an
  event $b$ such that $\bigwedge_{r}\Gop\Fop\vartheta_{r}$ holds iff
  $\Gop\Fop b$. Thus
  \[
    \theta \;\equiv\; \Fop\Gop\neg a \;\vee\; \Gop\Fop b .
  \]

  Take the product of the two counters and record, in each automaton state,
  whether the most recent transition produced $a$, $b$, both, or neither.
  Assign priority $0$ when $b$ occurs, priority $1$ when $a$ but not $b$
  occurs, and priority $2$ otherwise. If $b$ occurs infinitely often,
  priority $0$ is seen infinitely often. If $b$ occurs only finitely often
  but $a$ occurs infinitely often, the least priority seen infinitely often
  is $1$. If both events occur only finitely often, the run eventually sees
  only priority $2$. Hence the parity condition holds exactly when
  $\Fop\Gop\neg a \vee \Gop\Fop b$ holds.

  The two counters give $mk$ combinations, and recording the event adds only
  a constant factor. Transitions are computed by evaluating the relevant
  Boolean guards on the current letter, so the alphabet $2^{\AP}$ need not be
  enumerated. Details are in \appref{app:frag}.
\end{proof}

\subsection{Complexity classification}

\begin{thm}[One-dimensional fragments]\label{thm:frag-1d}
  Let the constraints be one-dimensional. Under both perfect-recall and
  finite-memory semantics, the following hold.
  \begin{enumerate}[\em(1)]
  \item Deciding a quantitative strategic atom $\coal{C}_{\Lambda}\top$ is
    polynomial-time equivalent to $\MPG$, and every formula of the
    quantitative fragment is decidable in $\NP\cap\coNP$.
  \item $\ATLmp$ atoms and flat $\ATLmp$ formulae are decidable in
    $\NP\cap\coNP$ and are $\MPG$-hard; model checking $\ATLmp$ is in
    $\PTIME^{\MPP}$. Each strategic subformula is handled in time polynomial
    in $\size{\game}$ and the numerical weight bound $\hW_{\varphi}$, and hence
    in pseudo-polynomial time.
  \item $\ATLsmpGR$ atoms and flat $\ATLsmpGR$ formulae are decidable in
    $\NP\cap\coNP$ and are $\MPG$-hard; model checking $\ATLsmpGR$ is in
    $\PTIME^{\MPP}$. Each strategic subformula is handled in time polynomial
    in $\size{\game}$, $\size{\varphi}$ and $\hW_{\varphi}$.
  \end{enumerate}
\end{thm}

\begin{proof}
  \emph{(1)} Taking $\psi = \top$ in \autoref{thm:seq} and
  \autoref{thm:fm-correspondence}, the atom holds at $s$ iff the threshold of $\Lambda$
  is enforceable in the one-dimensional mean-payoff game $\game^{C}$;
  perfect-recall and finite-memory winning coincide there because such games
  are memorylessly determined \cite{EhrenfeuchtM79}. The sequentialisation is
  polynomial by \autoref{def:seq}, and the threshold problem is in
  $\NP\cap\coNP$ \cite{ZwickP96}.

  Conversely, a turn-based mean-payoff game embeds as a WCGS by giving every
  inactive player a single action. For an edge-weighted input, first apply
  the edge-splitting construction of \appref{app:weights} and double the
  weights to preserve the threshold. The extension to the whole fragment
  follows by closure, since it is flat.

  \emph{(2)} \autoref{lem:atl-constructions} reduces each strategic atom to
  one mean-payoff parity game with at most two priorities, giving the
  $\NP\cap\coNP$ bound. Closure under Boolean combinations gives the bound
  for flat formulae. Evaluating nested formulae bottom-up requires at most
  one $\MPP$ query per strategic subformula, giving $\PTIME^{\MPP}$. Hardness
  is inherited from~(1) by taking $\chi = \Gop\top$.

  \emph{(3)} Instantiate the procedure of \autoref{sec:mc} with the automaton
  of \autoref{lem:gr1-dpa} in place of the general $\LTL$ determinisation.
  Since the automaton is kept symbolic and its counters are updated by
  evaluating the guards on $\lab(s)$, the product arena has size
  $O(\size{\game}^{2}\cdot m\cdot k)$ and three priorities, the
  one-dimensional game-solving bounds of \autoref{sec:mc-1d}---%
  \eqref{eq:mpp-time} for perfect recall and \autoref{lem:fm-energy} for
  finite memory---give a pseudo-polynomial running time and membership in
  $\NP\cap\coNP$. The bounds for atoms, flat formulae and nested model
  checking then follow as in item~(2). Hardness is inherited from the
  quantitative fragment by taking a tautological $\GR$ objective.
\end{proof}

Consequently the quantitative fragment is in $\PTIME$ if $\MPG$ is, and
both $\ATLmp$ and $\ATLsmpGR$ are in $\PTIME$ if $\MPP$ is.

These bounds retain the one-dimensional restriction and vary the temporal
component. We next lift that restriction and consider constraints over
several weight dimensions.

\FloatBarrier

\section{Multi-dimensional mean-payoff constraints}\label{sec:multi}
The main perfect-recall and finite-memory results above are one-dimensional.
We now allow $\Lambda$ to constrain several
weight dimensions simultaneously. Under memoryless semantics, the full logic
with arbitrary conjunctive constraints is already covered by
\autoref{thm:memoryless}. With conjunctive constraints over several
dimensions, memoryless, finite-memory and perfect-recall semantics yield
different answers even for pure mean-payoff objectives
\cite{VelnerCDHRR15}. We first classify pure quantitative
objectives under all three semantics, and then consider the full logic under
finite-memory semantics. The corresponding perfect-recall problem with a
non-trivial temporal objective remains open.

\subsection{Quantitative objectives}\label{sec:frag-multi}

When $\psi = \top$, the automaton $\dpa_{\top}$ has a single state with an
even priority, and the round-preserving product collapses to the coalition
sequentialisation $\game^{C}$. Writing
\[
  \Lambda = \bigwedge_{j \in J}\mpf_j \geq q_j
  \qquad\text{and}\qquad
  \vec{q} = (q_j)_{j \in J} ,
\]
deciding $\game, s \models \coal{C}_{\Lambda}\top$ is therefore exactly the
problem of whether player~$1$ can enforce the threshold vector $\vec{q}$ in
the multi-mean-payoff game $\game^{C}$, restricted to the dimensions
occurring in $\Lambda$. Under finite-memory semantics, player~$1$ is
correspondingly restricted to finite-memory strategies. Since
player~$1$ moves only at vertices corresponding to states of $\game$, its
positional strategies in $\game^{C}$ correspond exactly to memoryless
$C$-strategies in $\game$.

\begin{prop}[Multi-dimensional quantitative objectives]\label{prop:pure-multi}
  Let $\Lambda$ be an arbitrary conjunctive mean-payoff constraint. Deciding
  $\coal{C}_{\Lambda}\top$ is $\NP$-complete under memoryless semantics and
  $\coNP$-complete under both finite-memory and perfect-recall semantics.
  Under either of the latter two semantics, model checking flat formulae in
  the quantitative fragment is in $\PTIME^{\NP}_{\parallel}$.
\end{prop}

\begin{proof}
  Since the product reduces to the polynomial-size sequentialisation
  $\game^{C}$, the three atom bounds follow from the corresponding
  results for multi-dimensional mean-payoff games
  \cite{VelnerCDHRR15}; hardness transfers through the
  embedding used in \autoref{thm:frag-1d}(1).

  Under finite-memory and perfect-recall semantics each atom is decidable in
  $\coNP$. A polynomial-time machine can query the falsity of all atoms of a
  flat formula in parallel and then evaluate its Boolean structure, giving
  the $\PTIME^{\NP}_{\parallel}$ bound.
\end{proof}

Unless $\NP = \coNP$, the flat fragment is not contained in $\coNP$, since it
contains negations of $\coNP$-hard atoms. The identical $\coNP$ bounds under
finite-memory and perfect-recall semantics do not imply that the two
semantics agree: they already differ in one dimension, as \autoref{exa:sep}
shows.

\subsection{Full logic under finite memory}

Applying the $\coNP$ characterisation of explicit multi-dimensional energy
games to the doubly exponential product would give a triple-exponential
deterministic upper bound. Instead, we use the deterministic algorithm for
the arbitrary-initial-credit problem for multi-dimensional energy parity
games \cite[Corollary~V.5]{ColcombetJLS17}, obtaining the sharper $\TwoExp$
bound below.

\paragraph{From mean payoff to arbitrary initial credit.}
An energy objective treats the accumulated weight in each dimension as a
running balance. Given a play $\rho = v_0v_1\cdots$ of a weighted arena with
weight function $w$ and an initial-credit vector
$\vec{c} \in \mathbb{N}^{d}$, the energy condition requires
\[
  \vec{c} + \sum_{t=0}^{k-1} w(v_t) \;\geq\; \vec{0}
\]
componentwise for every $k \geq 0$. The
\emph{arbitrary-initial-credit problem} asks whether there exists an initial
credit $\vec{c}$ from which player~$1$ can enforce this condition together
with the parity objective. For example, if every reachable cycle has
non-negative weight but a path to such a cycle incurs a temporary loss of $5$
in one dimension, an initial credit of $5$ suffices to cover that loss;
repeating the cycle thereafter cannot decrease the balance further.

Under finite-memory strategies, this problem characterises conjunctive
mean-payoff parity objectives. Fixing a finite-memory player~$1$ strategy
yields a finite graph. The strategy guarantees non-negative mean payoff in
every dimension only if every reachable cycle has non-negative total weight
in every dimension: otherwise player~$2$ can reach and repeat a cycle with
negative weight in some dimension, yielding negative mean payoff in that
dimension and violating the corresponding constraint. Conversely, if every reachable cycle
has componentwise non-negative weight, then every prefix consists of such
cycles together with a residual simple path; the accumulated weight is
therefore bounded below, and some finite initial credit covers every
temporary loss.

Consequently, after normalising the thresholds to zero, each strategic
subformula under finite-memory semantics can be decided by solving an
arbitrary-initial-credit multi-dimensional energy parity game on the product
of \autoref{sec:seq}.

\begin{thm}\label{thm:multi-full}
  Model checking $\ATLsmp$ with arbitrary conjunctive mean-payoff constraints
  under finite-memory semantics is $\TwoExp$-complete.
\end{thm}

\begin{proof}
  Hardness follows from \autoref{thm:main} by taking every constraint to be
  $\top$.

  For the upper bound, evaluate the formula bottom-up as in
  \autoref{sec:mc}. Consider a strategic subformula
  $\coal{C}_{\Lambda}\psi$ after substitution of its maximal state
  subformulae, and form the round-preserving product
  $\Pi = \game^{C}\otimes\dpa_{\psi}$. By \autoref{thm:fm-correspondence},
  the subformula holds at $s$ iff player~$1$ has a finite-memory strategy
  winning $\mathrm{Par}\wedge\Lambda$ in $\Pi$ from $(s,q^{0})$.

  Normalise every threshold in $\Lambda$ to zero and move the state weights
  to the outgoing edges, which changes no mean payoff. We claim that
  player~$1$ has such a finite-memory strategy iff it wins the corresponding
  multi-dimensional energy parity game from some finite initial-credit vector.
  The reverse implication is immediate: a strategy maintaining non-negative
  energy from a finite initial credit has $\liminf$ mean payoff at least $0$
  in every dimension and satisfies the same parity condition. For the forward
  implication, fix a finite-memory winning strategy; the induced graph is finite, and every
  reachable cycle has non-negative total weight in every dimension
  $j \in \dm(\Lambda)$, since otherwise player~$2$ could reach and repeat a
  negative cycle, violating the corresponding mean-payoff constraint
  irrespective of whether the parity condition holds along that play. Every
  prefix decomposes into non-negative cycles and a residual simple path, so
  its accumulated weight is bounded below and some finite initial credit
  suffices. The parity condition is unchanged, since the strategy and its
  consistent paths are the same. This is the componentwise extension of the
  argument in \autoref{lem:fm-energy}; without parity, it is
  \cite[Theorem~3]{ChatterjeeDHR10}.

  It remains to bound the cost of solving that problem on $\Pi$. Let $N$ be
  the number of product vertices, $p$ the number of even priorities and $d$
  the number of dimensions; then
  $N \leq \size{\game}^{2}\cdot 2^{2^{O(\size{\varphi})}}$ and
  $p \leq 2^{O(\size{\varphi})}$. The arbitrary-initial-credit problem for
  multi-dimensional energy parity games is solvable deterministically in
  time
  \[
    (N \cdot \hW_{\Lambda})^{O((d+p)^{3}\log(d+p))}
  \]
  \cite[Corollary~V.5]{ColcombetJLS17}. Let $n$ be the combined size of
  $\game$ and $\varphi$. Then $d \leq n$, $p \leq 2^{O(n)}$,
  $\log N \leq 2^{O(n)}$ and $\log \hW_{\Lambda} \leq \mathrm{poly}(n)$, so the
  logarithm of the displayed bound is
  \[
    O\bigl((d+p)^{3}\log(d+p)\bigr)\cdot
      \bigl(\log N + \log \hW_{\Lambda}\bigr)
    \;=\; 2^{O(n)} .
  \]
  Hence the game can be solved in deterministic time $2^{2^{O(n)}}$.
  Repeating the computation for all product vertices and for all strategic
  subformulae preserves the bound.
\end{proof}

\subsection{Fixed number of dimensions}

The constructions of \autoref{sec:fragments} use only a bounded number of
priorities. When the number of weight dimensions is also fixed, the exponent
in the energy-parity algorithm becomes constant.

\begin{prop}[Fixed dimension]\label{prop:fixed-dim}
  Let the number $d$ of dimensions be fixed. Under finite-memory semantics, model
  checking the quantitative, $\ATLmp$ and $\ATLsmpGR$ fragments with
  conjunctive constraints over those dimensions is decidable in
  pseudo-polynomial time, that is, in time polynomial in $\size{\game}$,
  $\size{\varphi}$ and $\hW_{\varphi}$. In particular these problems are in
  $\PTIME$ when the weights and the numerators and denominators of the
  thresholds are written in unary.
\end{prop}

\begin{proof}
  Each fragment yields a game with a bounded number of priorities. The
  quantitative fragment has $\psi = \top$, so the product collapses to
  $\game^{C}$ and the parity condition is trivial. The constructions of
  \autoref{lem:atl-constructions} give at most two priorities on an arena of
  size $O(\size{\game}^{2})$; they do not touch the weights, and the safety
  restriction used for $\Gop\varphi$ relies only on prefix independence, so
  they apply verbatim to a $d$-dimensional constraint. The monitor of
  \autoref{lem:gr1-dpa} has three priorities and a product arena of size
  $O(\size{\game}^{2}\cdot m\cdot k)$.

  In each case, reduce to the arbitrary-initial-credit problem for
  multi-dimensional energy parity games, as in \autoref{thm:multi-full}. With $d$ and the number $p$ of
  even priorities both bounded, the exponent
  $O((d+p)^{3}\log(d+p))$ of \cite[Corollary~V.5]{ColcombetJLS17} is
  constant, so each game is solved in time polynomial in the arena size and
  in $\hW_{\varphi}$, and the model-checking procedure makes at most
  $\size{\varphi}$ such calls. Under a unary encoding $\hW_{\varphi}$ is
  polynomially bounded in the input size, so the running time is polynomial.
\end{proof}

\subsection{Perfect recall}

The argument of \autoref{thm:multi-full} fixes a finite-memory strategy and
reasons about the resulting finite graph. That step is unavailable for
perfect-recall strategies.

\begin{oprob}[Perfect recall in multiple dimensions]\label{oprob:multi-pr}
  The decidability and complexity of multi-dimensional mean-payoff parity
  games under perfect-recall strategies remain unknown. Without parity the
  problem is $\coNP$-complete \cite{VelnerCDHRR15}, as used in
  \autoref{prop:pure-multi}. A solution would yield the corresponding
  model-checking bounds for $\ATLsmp$ by \autoref{thm:seq}. Existing results
  cover finite-memory strategies with parity \cite{ChatterjeeRR14} and
  perfect-recall strategies without parity \cite{VelnerCDHRR15}, but not
  their combination.

The objective $\mathrm{Par}\wedge\Lambda$ is Borel, so the resulting
turn-based perfect-information games are determined \cite{Martin75}; what
is missing is an effective characterisation of the winning region. The
undecidability result of \cite{Velner15} concerns arbitrary Boolean
combinations of limit-average constraints, including disjunction and upper
bounds. Since these operators are excluded from
\autoref{def:constraint}, that result does not directly apply to the fragment
considered here.

\end{oprob}

\section{Memory requirements}\label{sec:memory}

We now study how the available strategy memory affects strategic ability.
Standard $\ATLs$ over finite concurrent game structures does not distinguish
perfect recall from finite memory: every enforceable $\omega$-regular
objective has a finite-memory witness \cite{AlurHK02,BullingJ14}. This
equivalence fails for $\ATLsmp$, even though mean-payoff constraints do not
increase the worst-case model-checking complexity (\autoref{thm:main}).

The three subsections have different scope. The monotonicity and
strict-hierarchy results apply to arbitrary conjunctive mean-payoff
constraints, with strictness already witnessed in one dimension. We then
establish finite-memory approximation for one-dimensional constraints
combined with a temporal objective; for pure quantitative objectives the same
approximation extends componentwise to several dimensions. Finally we give
tight bounds on the memory required by finite-memory witnesses, for
one-dimensional constraints.

\subsection{Strict memory hierarchy}

For positive formulae, enlarging the coalition's strategy class preserves
satisfaction, whereas enlarging the opponents' strategy class may only
destroy it. The argument is independent of the number of weight dimensions
and therefore applies to arbitrary conjunctive constraints $\Lambda$. Hence
\[
  \models_{\ML} \;\Longrightarrow\; \models_{\FM} \;\Longrightarrow\;
  \models .
\]
For the first implication, every memoryless coalition strategy is a
finite-memory strategy, while $\models_{\FM}$ quantifies only over
finite-memory opponents rather than all perfect-recall opponents. For the
second, \autoref{thm:fm-correspondence} shows that $\models^{\FM,\FM}$ agrees
with $\models^{\FM,\PR}$; enlarging the coalition strategy class from finite
memory to perfect recall then preserves satisfaction. These implications
extend from strategic atoms to positive formulae by structural induction. The
restriction to positive formulae is necessary because negation reverses both
directions. The following theorem shows that both implications
are strict.

\begin{thm}[Memory hierarchy]\label{thm:sep}
  There is a two-state, one-player WCGS $\game$ with a single weight
  dimension, a state $s$, and two formulae $\varphi_{2/3}, \varphi_1$ of
  $\ATLsmpGR$ with one-dimensional constraints, such that
  \begin{align*}
    \game, s \models_{\FM} \varphi_{2/3}
      &\qquad\text{and}\qquad
    \game, s \not\models_{\ML} \varphi_{2/3} , \\
    \game, s \models \varphi_{1}
      &\qquad\text{and}\qquad
    \game, s \not\models_{\FM} \varphi_{1} .
  \end{align*}
  Thus, on positive $\ATLsmp$ formulae with arbitrary conjunctive
  constraints,
  \[
    \models_{\ML} \;\Longrightarrow\; \models_{\FM} \;\Longrightarrow\;
    \models ,
  \]
  and both implications are strict, already for one-dimensional
  $\ATLsmpGR$ formulae over a fixed two-state game.
\end{thm}

\begin{exa}[Separating game]\label{exa:sep}
  Let $N = C = \{1\}$, $d = 1$, $\AP = \{p\}$, and let $\game$ have two
  states $s$ and $u$ with $\lab(s) = \{p\}$, $\lab(u) = \emptyset$,
  $w(s) = 0$ and $w(u) = 1$. Agent~$1$ has
  $\Av_1(s) = \Av_1(u) = \{\mathit{stay}, \mathit{go}\}$. At either state,
  $\mathit{stay}$ is a self-loop and $\mathit{go}$ moves to the other state.
  For
  $q \in \mathbb{Q}$ write
  $\varphi_q = \coal{\{1\}}_{\mpf_1 \geq q}\,\Gop\Fop p$.
\end{exa}

\begin{proof}[Proof of \autoref{thm:sep}]
  We use the game of \autoref{exa:sep} for both separations. Since $C = N$,
  every strategy determines a unique play, and the opponent strategy class is
  irrelevant.

  \emph{Memoryless is weaker than finite memory.} A memoryless strategy is a
  choice of action at $s$ and a choice at $u$. If it plays $\mathit{stay}$ at
  $s$, its choice at $u$ is irrelevant and the outcome is $s^{\omega}$, which
  satisfies $\Gop\Fop p$ but has mean payoff $0$. If it plays $\mathit{go}$
  at $s$ and $\mathit{stay}$ at $u$, the play is $s\,u^{\omega}$, which fails
  $\Gop\Fop p$. If it plays $\mathit{go}$ at both, the play is
  $(su)^{\omega}$, which satisfies $\Gop\Fop p$ with mean payoff $\tfrac12$.
  These cases are exhaustive, so the maximum mean payoff attainable by a
  memoryless strategy satisfying $\Gop\Fop p$ is $\tfrac12$, and
  $\game, s \not\models_{\ML}\varphi_{2/3}$. A two-state transducer, on the
  other hand, realises the play $(suu)^{\omega}$: it plays $\mathit{go}$ at
  $s$, then $\mathit{stay}$ on the first visit to $u$ and $\mathit{go}$ on
  the second, its memory recording whether the current visit to $u$ is the
  first or the second since the last visit to $s$. This satisfies
  $\Gop\Fop p$ and has mean payoff $\tfrac23$, so
  $\game, s \models_{\FM}\varphi_{2/3}$.

  \emph{Finite memory is weaker than perfect recall.} On the $i$-th visit to
  $s$, for $i \geq 1$, the strategy moves to $u$, remains there for $2^{i}$
  positions, and then returns to $s$. The resulting play visits $s$
  infinitely often, so $\Gop\Fop p$ holds. The number of visits to $s$ grows
  only linearly, while the total play length grows exponentially, so the
  asymptotic density of the weight-$0$ state tends to $0$. Concretely, after
  the $i$-th excursion the play has length $T_i = i + 2^{i+1} - 2$ and
  accumulated weight $2^{i+1}-2$, so the average over that prefix is
  $1 - i/T_i \to 1$. Between successive excursion endpoints the average can
  decrease only at the single visit to $s$, which contributes weight $0$;
  immediately after that visit the average is $1 - (i{+}1)/(T_i{+}1)$, and it
  increases again for the rest of the excursion. The size of the decrease
  therefore tends to $0$, so $\mpf_1 = 1$ and
  $\game, s \models \varphi_1$.

  Conversely, let $\sigma_1$ be a finite-memory strategy of size $m$. Since
  the game and the transducer are deterministic, the induced sequence of
  pairs consisting of a game state and a memory state is ultimately
  periodic, and so therefore is the induced play $\out$: writing
  $\out = h\,\lambda^{\omega}$, we have
  $|\lambda| = \ell \leq 2m$. If $\out \models \Gop\Fop p$ then $s$ occurs in
  $\lambda$, say $c \geq 1$ times, and by prefix-independence
  \[
    \mpf_1(\out) \;=\; \frac{\ell - c}{\ell} \;\leq\; 1 - \frac{1}{\ell}
    \;<\; 1 ,
  \]
  so the constraint $\mpf_1 \geq 1$ fails. If $s$ does not occur in
  $\lambda$, then $\Gop\Fop p$ fails instead. Either way no finite-memory
  strategy witnesses $\varphi_1$.
\end{proof}

The two strict inclusions have different causes. A memoryless strategy must
use the same action whenever the same state is revisited, and therefore
cannot vary the length of successive excursions to \(u\). A finite-memory
strategy can realise longer periodic excursions, but every induced play is
ultimately periodic; it cannot make the excursion lengths grow without
bound, as required to attain mean payoff \(1\) while still visiting \(s\)
infinitely often.

The difference between finite memory and perfect recall is therefore confined
to the boundary threshold. In \autoref{exa:sep}, \(\varphi_q\) holds under
both \(\models\) and \(\models_{\FM}\) for every \(q<1\): choosing
\(\ell\) sufficiently large and visiting \(s\) once every \(\ell\) positions
gives mean payoff \(1-1/\ell\geq q\). Both semantics fail for \(q>1\), while
only perfect recall satisfies \(\varphi_1\). The next theorem shows that this
approximation from below holds generally in one dimension.

\subsection{Finite-memory approximation}

The following result lifts the finite-memory approximation property of
mean-payoff parity games \cite{ChatterjeeHJ05,ChatterjeeD12} to coalition
abilities in weighted concurrent games, through the strategy correspondence
of \autoref{sec:seq}.

\begin{thm}[Finite-memory approximation]\label{thm:approx}
  Let $\game$ be a WCGS, $s$ a state, $C$ a coalition, $\psi$ an $\LTL$
  formula and $j$ a dimension. For all rationals $q < q'$,
  \[
    \game, s \models \coal{C}_{\mpf_j \geq q'}\psi
    \qquad\Longrightarrow\qquad
    \game, s \models_{\FM} \coal{C}_{\mpf_j \geq q}\psi .
  \]
\end{thm}
\begin{proof}
  By \autoref{thm:seq}, player~$1$ can enforce
  $\mathrm{Par}\wedge(\mpf \geq q')$ in $\Pi = \game^{C}\otimes\dpa_{\psi}$
  from $(s,q^{0})$. For every $\delta > 0$, player~$1$ then has a
  finite-memory strategy enforcing the parity condition and mean payoff at
  least $q' - \delta$ \cite[proof of Theorem~4]{ChatterjeeD12}, a property
  established for mean-payoff parity games in \cite{ChatterjeeHJ05}. Choose
  $0 < \delta < q' - q$; the resulting strategy wins for
  $\mathrm{Par}\wedge(\mpf\geq q)$, and
  \autoref{thm:fm-correspondence}(2) transfers it back to a finite-memory
  $C$-strategy witnessing the conclusion.
\end{proof}

Consequently, fix $\game$, $s$, $C$, $\psi$ and $j$, and let
\[
v=\sup\{q\in\mathbb{Q}: \game,s\models\coal{C}_{\mpf_j\geq q}\psi\}.
\]
The same value is obtained under finite-memory semantics. Indeed, every
finite-memory strategy is a perfect-recall strategy, so every threshold
enforceable with finite memory is also enforceable with perfect recall.
Conversely, let $q<v$ be rational. By the definition of $v$, there is a
rational $q'$ with $q<q'$ that is enforceable with perfect recall.
By \autoref{thm:approx}, $q$ is then enforceable with finite memory.

Thus the two semantics agree on every rational threshold strictly below
$v$. They may differ only in whether $v$ itself is attained, when
$v\in\mathbb{Q}$: in \autoref{thm:sep}, for example, both semantics enforce
every $q<1$, but only perfect recall enforces the boundary threshold $1$.

\autoref{thm:approx} is stated for a one-dimensional constraint because its
proof uses the finite-memory approximation property of one-dimensional
mean-payoff parity games. For pure quantitative objectives the statement
extends componentwise. If player~$1$ wins a multi-dimensional mean-payoff
game for the threshold vector $\vec{0}$, then for every $\alpha > 0$ it has a
finite-memory strategy securing at least $-\alpha$ in every dimension
\cite[Lemma~15]{VelnerCDHRR15}. Taking $\psi = \top$, so that the product
collapses to $\game^{C}$ as in \autoref{sec:frag-multi}, and normalising the
thresholds, this gives
\[
  \game, s \models \coal{C}_{\bigwedge_{j \in J}\mpf_j \geq q'_j}\top
  \qquad\Longrightarrow\qquad
  \game, s \models_{\FM} \coal{C}_{\bigwedge_{j \in J}\mpf_j \geq q_j}\top
\]
for all rational vectors $\vec{q} < \vec{q}\,'$; choose
$0 < \alpha \leq \min_{j}(q'_j - q_j)$ and transfer the strategy back with
\autoref{thm:fm-correspondence}(2). Consequently the perfect-recall and
finite-memory threshold regions of a quantitative strategic atom have the
same downward closure, and can differ only on their Pareto boundary. We do
not establish the corresponding statement for multi-dimensional constraints
combined with a non-trivial temporal objective; that case would need the
approximation property for multi-dimensional mean-payoff parity games, which
is also the setting of \autoref{oprob:multi-pr}.

\subsection{One-dimensional memory bounds}

Constraints are one-dimensional throughout this subsection. Even when the
game and temporal objective are fixed, the memory required by a
finite-memory witness may depend on the numerical threshold. As in
\autoref{sec:mc}, a one-dimensional conjunction is equivalent to its
strongest conjunct, so it suffices to consider constraints of the form
$\mpf_j \geq q$.

\begin{thm}[Tight memory bounds]\label{thm:memory-bounds}
  The following hold.
  \begin{enumerate}[\em(1)]
  \item Let $\game$ be the fixed two-state game of \autoref{exa:sep} and, for
    every integer $b \geq 2$, put
    \[
      \varphi_b \;=\; \coal{\{1\}}_{\mpf_1 \geq 1-1/b}\,\Gop\Fop p .
    \]
    Then $\game, s \models_{\FM}\varphi_b$, but every witnessing
    finite-memory strategy has at least $\lceil b/2 \rceil$ memory states.
    Thus the required memory may be exponential in the binary encoding of the
    threshold, even though the game and temporal objective are fixed.
  \item Suppose $\game, s \models_{\FM}\coal{C}_{\mpf_j \geq q}\psi$, where
    $q = a/b$ is in lowest terms, and let $\dpa_{\psi}$ have state space $Q$
    and $c$ priorities. Then there is a witnessing finite-memory
    $C$-strategy with
    \[
      O\bigl(\size{\game}^{2}\cdot|Q|^{2}\cdot c\cdot W_q\bigr)
    \]
    memory states, where $W_q = bW + |a|$.
  \end{enumerate}
\end{thm}

\begin{proof}
  \emph{Lower bound.} A $b$-state counter realises the periodic play
  $(s\,u^{b-1})^{\omega}$. This play satisfies $\Gop\Fop p$ and has mean
  payoff $(b-1)/b = 1-1/b$, so it witnesses $\varphi_b$.

  Conversely, let a finite-memory witness have $m$ memory states. Since the
  game and the strategy are deterministic, the induced play is ultimately
  periodic, with a period $\lambda$ of length $\ell \leq 2m$. The objective
  $\Gop\Fop p$ requires $s$ to occur in $\lambda$; if it occurs $r \geq 1$
  times, then
  \[
    \mpf_1 \;=\; \frac{\ell - r}{\ell} \;\leq\; 1 - \frac{1}{\ell} .
  \]
  Satisfying the threshold $1-1/b$ therefore requires $\ell \geq b$. Hence
  $2m \geq b$, and thus $m \geq \lceil b/2 \rceil$.

  \emph{Upper bound.} By \autoref{thm:fm-correspondence}, player~$1$ has a
  finite-memory strategy winning $\mathrm{Par}\wedge(\mpf \geq q)$ in the
  product $\Pi = \game^{C}\otimes\dpa_{\psi}$. This product has
  $|V| = O(\size{\game}^{2}|Q|)$ vertices and $c$ priorities. Normalising the
  threshold replaces the weights by $b\cdot w - a$, whose absolute value is
  at most $W_q$. By \autoref{lem:fm-energy}, player~$1$ therefore has a
  winning strategy in $\Pi$ with
  $O(|V|\cdot c\cdot W_q) = O(\size{\game}^{2}|Q|\cdot c\cdot W_q)$ memory
  states. Translating this strategy back through
  \autoref{thm:fm-correspondence} contributes a further factor of $|Q|$,
  giving the claimed bound.

  Each member of $C$ may implement a private copy of the resulting joint
  transducer and project its joint output to its own action.
\end{proof}

Any satisfiable non-trivial threshold may be taken in $[-W,W]$: thresholds
below $-W$ are vacuous, while thresholds above $W$ are unenforceable. For
$q = a/b \in [-W,W]$ we have $|a| \leq bW$, and hence $W_q \leq 2bW$.
Consequently, when the game and temporal monitor are fixed, the upper bound
is $O(b)$, matching the $\Omega(b)$ lower bound in item~\emph{(1)}. Thus the
worst-case memory requirement is $\Theta(b)$ and may be exponential in the
binary size of the threshold.

\section{Applications and expressive limits}\label{sec:apps}

This section gives two direct applications of the quantitative modality:
temporal synthesis with performance guarantees and multi-criteria coalition
objectives. We then compare the logic with cooperative rational verification.
The comparison exposes an expressive limitation: the logic can impose fixed
payoff thresholds, but cannot compare the payoffs induced by two strategy
profiles.

\subsection{Synthesis and optimisation}

Let $C$ be the controlled components and $N \setminus C$ the environment.
Then $\coal{C}_{\Lambda}\psi$ asks whether the controller can enforce the
temporal specification $\psi$ while guaranteeing the long-run performance
bounds $\Lambda$ against every environment strategy. Taking
$\Lambda = \top$ gives qualitative temporal synthesis over the game
structure, while taking $\psi = \top$ gives a multi-mean-payoff game. When
both are non-trivial, this captures reactive-controller synthesis in which
safety or liveness requirements must coexist with long-run performance
guarantees, for example on resource use, throughput or service reward. It is
the coalition-indexed concurrent-game counterpart of the
qualitative--quantitative synthesis problems studied in
\cite{BloemCHJ09,BohyBFJR13}.

Repeated threshold queries also yield an optimisation procedure. For fixed
$\game$, $s$, $C$, $\psi$ and a dimension $j$, define
\[
  v_C(s,\psi,j) \;=\;
  \sup\{ q \in \mathbb{Q} :
    \game, s \models \coal{C}_{\mpf_j \geq q}\,\psi \} ,
\]
with $\sup\emptyset = -\infty$. If the set is non-empty, it is downward
closed and its supremum lies in $[-W,W]$. Bisection on this interval, using
the procedure of \autoref{thm:main} as a threshold oracle, therefore
approximates $v_C(s,\psi,j)$ to any prescribed additive error. By
\autoref{thm:approx}, perfect-recall and finite-memory semantics yield the
same supremum, although they may differ in whether it is attained;
\autoref{exa:sep} gives such a case.

\subsection{Aggregate and multi-criteria guarantees}

Dimensions need not represent individual agents. They may also record
system-level quantities, such as total reward, energy consumption or
throughput \cite{ChatterjeeRR14,VelnerCDHRR15}. In the player-indexed case
$d = n$, dimension $i$ represents the utility of agent $i$. The sum and
minimum of the individual utilities give utilitarian and egalitarian
social-welfare criteria, respectively; these forms of aggregation have been
considered in quantitative rational synthesis and temporal equilibrium
analysis \cite{AlmagorKP18,GutierrezNPW20}.

An egalitarian guarantee is then expressed directly by
\[
  \coal{C}_{\bigwedge_{i \in N}\mpf_i \geq q}\,\psi ,
\]
which requires $C$ to enforce $\psi$ while guaranteeing every agent a payoff
of at least $q$. For a
utilitarian guarantee, add a dimension
$w_{\mathit{util}}(s) = \sum_{i \in N} w_i(s)$ and require
$\mpf_{\mathit{util}} \geq q$. Since the lower mean payoff is superadditive,
\[
  \mpf_{\mathit{util}}(\out) \;\geq\; \sum_{i \in N}\mpf_i(\out) ,
\]
with equality when all individual running averages converge. Separate lower
bounds therefore imply a corresponding aggregate bound, but the converse need
not hold: high payoff in one dimension may compensate for low payoff in
another.

For example, suppose that a group of service robots must repeatedly complete
a collection of tasks. An egalitarian guarantee requires every robot to
receive a long-run average net payment of at least $q$, ruling out a policy
that achieves high aggregate performance by systematically under-rewarding
one robot. A utilitarian guarantee instead requires the robots' combined
long-run utility to exceed a given threshold, allowing gains to one robot to
compensate for losses to another. Aggregate dimensions can also express
system constraints. If $p_i(s)$ is the payment made to robot $i$ at state
$s$, defining
\[
  w_{\mathit{bud}}(s) \;=\; -\sum_{i \in N} p_i(s)
\]
and requiring $\mpf_{\mathit{bud}} \geq -B$ ensures that the long-run average
total payment does not exceed $B$, while the temporal objective $\psi$
requires the tasks themselves to be completed. We call this an
\emph{aggregate budget} guarantee rather than a utilitarian one: utilitarian
welfare aggregates the agents' utilities, whereas the budget dimension
measures the system's expenditure. Note that, since mean payoff ignores
finite prefixes, this is a constraint on the long-run rate of payment and not
on a finite cumulative budget.

For fixed $C$, $s$ and $\psi$, let
\[
  V_C(s,\psi) \;=\;
  \Bigl\{ \vec{x} \in \mathbb{Q}^{d} \;:\;
    \game, s \models
      \coal{C}_{\bigwedge_{j \in D}\mpf_j \geq x_j}\,\psi \Bigr\} .
\]
This set is downward closed: a strategy enforcing $\vec{x}$ also enforces
every $\vec{x}\,' \leq \vec{x}$. Its maximal points describe the trade-offs
among the dimensions; \autoref{exa:service} illustrates this in two
dimensions. For $\psi = \top$, the corresponding Pareto problem is studied in
\cite{BrenguierR15}. With a non-trivial temporal objective,
\autoref{thm:multi-full} gives membership queries under finite-memory
semantics, while the corresponding perfect-recall problem remains open
(\autoref{oprob:multi-pr}). We do not consider the computation of the Pareto
boundary itself.

\subsection{Cooperative rational verification}\label{sec:rv}

Cooperative rational verification asks which temporal properties hold in
profiles from which no coalition has a beneficial deviation
\cite{GutierrezKW23,GutierrezLNSW24}. Since a beneficial deviation is itself
a coalition-ability query, this suggests a connection with $\ATLsmp$. For
dichotomous preferences, the standard $\ATLs$ encoding of core existence uses
an improvement threshold fixed by the objectives. Under mean-payoff
preferences, the threshold instead depends on the payoff of the candidate
profile. Consequently $\ATLsmp$ expresses deviations from a fixed payoff
baseline, but does not directly reproduce the core encoding.

The \emph{core} originates in cooperative game theory
\cite{Gillies59,OsborneR94}. In the concurrent-game formulation of
\cite{GutierrezKW23}, a strategy profile is in the core if no coalition has a
deviation making every member of that coalition strictly better off, whatever
the agents outside it do. With dichotomous preferences
the core is always non-empty \cite[Theorem~1]{GutierrezKW23}, whereas under
mean-payoff preferences it may be empty \cite{GutierrezLNSW24}. We follow
\cite{GutierrezLNSW24} and specialise to the player-indexed case $d = n$,
where $w_i$ is the utility of agent $i$ and
\[
  \pay_i(\vec{\sigma}) \;=\; \mpf_i(\out(\vec{\sigma}, \sinit)) .
\]
We use finite-memory semantics throughout the subsection.

\paragraph{Dichotomous preferences.}
Let $\gamma_i$ be an $\LTL$ objective for agent $i$, with payoff $1$ when
$\gamma_i$ holds along the outcome and $0$ otherwise. Strict improvement is
therefore equivalent to changing $\gamma_i$ from false to true. A coalition
$C$ has a beneficial deviation precisely when every member of $C$ currently
loses and $C$ can jointly enforce all of their objectives. The latter
condition is expressed by
$\coal{C}\bigwedge_{i \in C}\gamma_i$
\cite[Theorem~2]{GutierrezKW23}.

Gutierrez et al.\ \cite[Theorem~3]{GutierrezKW23} use this observation to
encode core existence. Writing $N$ for the set of agents, the existence of a core profile
whose outcome satisfies a given $\LTL$ formula $\varphi$ is expressed by
\[
  \Phi_{\Core}(\varphi) \;=\;
  \bigvee_{W \subseteq N}
  \Bigl(
    \coal{N}\bigl(\varphi \wedge
      \textstyle\bigwedge_{i \in W}\gamma_i \wedge
      \bigwedge_{j \in N\setminus W}\neg\gamma_j \bigr)
    \;\wedge
    \bigwedge_{\emptyset \neq L \subseteq N\setminus W}
      \neg\coal{L}\textstyle\bigwedge_{j \in L}\gamma_j
  \Bigr) .
\]
Each disjunct fixes a set $W$ of winners. The grand-coalition modality
selects a profile whose outcome satisfies $\varphi$ and whose winners are
exactly $W$. The remaining conjuncts exclude a beneficial deviation by every
coalition of losers. Core non-emptiness is the case $\varphi = \top$.

The encoding works because the deviation test
$\coal{L}\bigwedge_{j \in L}\gamma_j$ is independent of the profile selected
by the outer modality: the winning threshold is fixed by the objectives
themselves. The formula expresses the existence of a core profile. It does
not express membership of an externally supplied profile, since neither
$\ATLs$ nor $\ATLsmp$ can name that profile.

\paragraph{Mean-payoff preferences.}
Define $x_i = \pay_i(\vec{\sigma})$ for $i \in C$. With mean-payoff
preferences, a coalition $C$ has a beneficial deviation from $\vec{\sigma}$
when it has a $C$-strategy $\vec{\sigma}\,'_C$ with
\[
  \pay_i\bigl(\out((\vec{\sigma}\,'_C, \vec{\tau}_{-C}), \sinit)\bigr)
  \;>\; x_i
  \qquad\text{for every } i \in C
\]
and every counter-strategy $\vec{\tau}_{-C}$.

The baseline $\vec{x}$ is the payoff vector of $\vec{\sigma}$, which varies
with the profile the outer modality selects. An encoding in the style of
$\Phi_{\Core}$ would therefore need the thresholds of an inner modality to
depend on the payoff of the outer profile. The syntax of $\ATLsmp$ provides
only fixed rational thresholds: it cannot name a strategy profile or compare
the payoffs of two outcomes. This rules out the direct encoding; it does not
establish that core non-emptiness is inexpressible in $\ATLsmp$.

\paragraph{Fixed-baseline deviations.}
Once the baseline is fixed, the deviation test is expressible: under
finite-memory semantics, a strict guarantee above the baseline is equivalent
to a non-strict guarantee at some larger rational threshold.

\begin{prop}[Fixed-baseline deviations]\label{prop:fixed-baseline}
  Let $C \subseteq N$, $s \in \St$, $\psi$ an $\LTL$ formula and
  $\vec{x} \in \mathbb{Q}^{C}$. The following are equivalent.
  \begin{enumerate}
    \item There is a finite-memory $C$-strategy that enforces $\psi$ and
      guarantees every member $i \in C$ a payoff strictly greater than $x_i$
      against every counter-strategy.
    \item There is a rational vector $\vec{x}\,' > \vec{x}$ such that
      \[
        \game, s \models_{\FM}
          \coal{C}_{\bigwedge_{i \in C}\mpf_i \geq x'_i}\psi .
      \]
  \end{enumerate}
\end{prop}

\begin{proof}
  Item~(2) implies item~(1): we have $\vec{x}\,' > \vec{x}$, and by
  \autoref{thm:fm-correspondence} a finite-memory witness against all
  finite-memory counter-strategies withstands all counter-strategies.

  Conversely, fix a finite-memory $C$-strategy $\vec{\sigma}_C$ satisfying
  item~(1). By \autoref{thm:fm-correspondence}(1) it induces a finite-memory
  player~$1$ strategy in $\game^{C}\otimes\dpa_{\psi}$. Let $H$ be the finite
  graph obtained by fixing that strategy. The projections of the infinite
  paths of $H$ from $(s,q^{0})$ are exactly the outcomes consistent with
  $\vec{\sigma}_C$, and every such path satisfies the parity condition of
  $\dpa_{\psi}$.

  Fix $i \in C$ and let $y_i$ be the minimum mean weight in dimension $i$
  among the reachable simple cycles of $H$. Since every simple cycle has
  length at most $|H|$, this is a minimum over finitely many rationals with
  denominator at most $|H|$. Every finite prefix of a path of $H$ decomposes
  into simple cycles together with a residual simple path, whose total weight
  is bounded by $|H|W$ independently of the length of the prefix. As each
  simple cycle has mean at least $y_i$, every infinite path of $H$ has lower
  mean payoff at least $y_i$. Conversely, a lasso reaching and repeating a
  simple cycle attaining $y_i$ has mean payoff exactly $y_i$ in dimension
  $i$, and is the outcome of some counter-strategy. By item~(1) that payoff
  is strictly greater than $x_i$, so $y_i > x_i$. Taking $x'_i = y_i$ for
  every $i \in C$ proves item~(2).
\end{proof}

For fixed $C$, a fixed vector $\vec{x}\,'$ and $\psi = \top$, the formula in
item~(2) is a quantitative strategic atom, whose model-checking problem is
$\coNP$-complete by \autoref{prop:pure-multi}. The fixed-baseline deviation
problem additionally existentially quantifies over a rational vector
$\vec{x}\,' > \vec{x}$, which the syntax of $\ATLsmp$ cannot do.

The problem \Dominated\ of \cite{GutierrezLNSW24} existentially chooses the
deviating coalition as well as the improved payoff vector. Coalitions could
be enumerated by a disjunction over the $2^{n}-1$ non-empty subsets of $N$,
at exponential cost in formula size, but $\ATLsmp$ has no quantifier over
rational threshold vectors. \autoref{prop:fixed-baseline} therefore captures
deviations from a fixed coalition and a fixed payoff baseline. Core existence
additionally quantifies over a candidate profile and binds the deviation
thresholds to the payoff of that profile. A direct encoding would require
strategy variables and bindings, as provided by Strategy Logic
\cite{ChatterjeeHP10,MogaveroMPV14}, extended with payoff terms referring to
the outcomes induced by named strategy profiles.

\section{Concluding remarks}\label{sec:conclusion}

We have introduced $\ATLsmp$, in which a strategic modality requires one
coalition strategy to satisfy a temporal objective and a mean-payoff
guarantee simultaneously. This combined ability is not reducible to separate
qualitative and quantitative modalities (\autoref{prop:nondecomp}).

A round-preserving sequentialisation reduces each strategic subformula to a
mean-payoff parity game while ensuring that the temporal automaton advances
once per original concurrent round (\autoref{thm:seq}). With one-dimensional
constraints, model checking is $\TwoExp$-complete (\autoref{thm:main}), matching $\ATLs$;
the upper bound is driven by $\LTL$ determinisation. With arbitrary
conjunctive constraints, it remains $\TwoExp$-complete under finite-memory
semantics (\autoref{thm:multi-full}). Memoryless,
finite-memory and perfect-recall abilities form a strict hierarchy
(\autoref{thm:sep}), while finite memory achieves every threshold strictly
below the perfect-recall supremum (\autoref{thm:approx}). The required memory
may be exponential in the binary size of the threshold
(\autoref{thm:memory-bounds}).

For dichotomous $\LTL$ preferences, $\ATLs$ can encode core existence because
the improvement test is fixed independently of the candidate profile. Under
mean-payoff preferences, the deviation threshold is the payoff of that
profile, which $\ATLsmp$ cannot name. The logic nevertheless expresses
deviations from a fixed payoff baseline (\autoref{prop:fixed-baseline}).

The main open question is the decidability of multi-dimensional mean-payoff
parity games under perfect-recall strategies (\autoref{oprob:multi-pr}). The
strategy correspondence also relies on perfect information and on
deterministic strategies and transitions (\autoref{rem:quantifier}); it does
not directly extend to imperfect-information, randomised or stochastic
models.

\section*{Acknowledgement}

\noindent I thank Julian Gutierrez, Anthony~W.~Lin, Thomas Steeples and
Michael~Wooldridge for the joint work on cooperative concurrent mean-payoff games from
which this paper grew.

\bibliographystyle{alphaurl}
\bibliography{refs}

\clearpage
\appendix

\section{State and transition weights}
\label{app:weights}

\begin{rem}[State weights versus transition weights]
  We attach weights to states rather than to transitions. State and
  transition weights are not interchangeable without adjusting the
  construction. The usual
  encoding of a transition-weighted game splits each transition $s \to s'$
  through a fresh state carrying the weight of the transition, with the
  original states given weight $\vec{0}$. A play of the original game of
  length $n$ then becomes a play of length $2n$ carrying the same total
  weight, so the mean payoff is \emph{halved}, not preserved; thresholds must
  be scaled accordingly, or, equivalently, the transition weights doubled.
  Moreover the inserted states introduce stuttering and therefore do not
  preserve arbitrary $\LTL$ formulae containing $\Xop$. The scaling can be
  corrected by doubling the transition weights, or equivalently by scaling
  the thresholds; the stuttering can be handled by updating the temporal
  automaton only at the original states, as in the round-preserving product
  of \autoref{sec:seq}. We avoid both adjustments by taking state weights as
  primitive.
\end{rem}

\section{Proofs for the strategy correspondence}
\label{app:seq}

\begin{proof}[Proof of \autoref{prop:projection}]
  We argue the three claims separately. \emph{Projection.} By
  \autoref{def:seq}, each pair of consecutive edges
  $s^{t} \to (s^{t},\alpha^{t}) \to s^{t+1}$ of a play $\rho$ witnesses the
  existence of $\beta^{t} \in \vec{\Ac}_{-C}(s^{t})$ with
  $\tr(s^{t},(\alpha^{t},\beta^{t})) = s^{t+1}$, so every consecutive pair of
  $\proj(\rho)$ is a legal transition of $\game$.

  \emph{Surjectivity.} Given a play $\out = s^{0}s^{1}\cdots$ of $\game$,
  choose
  for each $t$ an action profile $\acv^{t} \in \vec{\Ac}(s^{t})$ with
  $\tr(s^{t},\acv^{t}) = s^{t+1}$ and put $\alpha^{t} = \acv^{t}_{C}$; then
  $s^{0}(s^{0},\alpha^{0})s^{1}\cdots$ is a play of $\game^{C}$ projecting to
  $\out$. If a player~$1$ strategy $\sigma_1$ is fixed, the $V_2$-vertices of
  a play consistent with $\sigma_1$ are determined by $\sigma_1$ applied to
  the prefixes already built, so the preimage is unique by induction on
  $t$.

  \emph{Mean payoff.} Write $\out = \proj(\rho) = s^{0}s^{1}\cdots$. By
  construction the weight sequence of $\rho$ is
  $w_j(s^{0})\,w_j(s^{0})\,w_j(s^{1})\,w_j(s^{1})\cdots$, each value repeated
  twice, so for an even prefix length $2n$,
  \[
    \frac{1}{2n}\sum_{t < 2n} \hat{w}_j(\rho[t])
    \;=\; \frac{1}{2n}\cdot 2\sum_{t<n} w_j(s^{t})
    \;=\; \frac{1}{n}\sum_{t<n} w_j(s^{t}) ,
  \]
  and for an odd prefix length the two quantities differ by at most
  $2W/(2n{+}1)$, which tends to $0$. Hence the two sequences of running
  averages have the same $\liminf$.
\end{proof}

\begin{proof}[Proof of \autoref{lem:projection}]
  (1) By induction on $t$. The $U_1$-vertices of $\bar{\rho}$ are
  $(s^{0},q^{0}), (s^{1},q^{1}), \dots$ where, by \autoref{def:rpp},
  $q^{t+1} = \delta(q^{t}, \lab(s^{t}))$; this is the defining recursion of
  the run of $\dpa_{\psi}$ on $\lab(s^{0})\lab(s^{1})\cdots = \lab(\out)$.

  (2) By (1), the set of automaton states visited infinitely often along
  $\bar{\rho}$ coincides with the set visited infinitely often by the run of
  $\dpa_{\psi}$ on $\lab(\out)$; the passage from $U_1$- to $U_2$-vertices
  duplicates each priority but introduces and removes none. Hence the least
  priority occurring infinitely often is the same on both sides, and since
  $\dpa_{\psi}$ is deterministic and recognises exactly the models of $\psi$,
  the two conditions agree.

  (3) The weight of a product vertex depends only on its $\St$-component, so
  $\bar{w}_j(\bar{\rho}) = \hat{w}_j(\rho)$ and \autoref{prop:projection}
  applies.
\end{proof}

\begin{proof}[Proof of \autoref{thm:seq}]
  Write $\Pi = \game^{C} \otimes \dpa_{\psi}$ and let $\mathrm{Obj}$ denote
  $\mathrm{Par}\wedge\Lambda$. We write $\proj$ also for the composite
  projection from plays of $\Pi$ to $\St^{\omega}$. By
  \autoref{lem:projection}, a play $\bar{\rho}$ of $\Pi$ lies in
  $\mathrm{Obj}$ iff $\proj(\bar{\rho}) \models \psi \wedge \Lambda$.

  \smallskip\noindent\emph{From $\game$ to $\Pi$.}
  Let $\vec{\sigma}_C$ be a $C$-strategy such that every counter-strategy
  yields an outcome satisfying $\psi$ and $\Lambda$. Define $\sigma_1$ in
  $\Pi$ as follows: for a history $\bar{h}$ of $\Pi$ ending in a
  $U_1$-vertex with automaton component $q$ and state component $s'$, put
  $\sigma_1(\bar{h}) = ((s',\alpha), q)$ where
  $\alpha = \vec{\sigma}_C(\proj(\bar{h}))$. This is a legal move, since
  $\alpha \in \vec{\Ac}_C(s')$.

  Let $\sigma_2$ be any player~$2$ strategy, let $\bar{\rho}$ be the
  resulting play from $(s,q^{0})$ and put $\out = \proj(\bar{\rho})$. By
  \autoref{prop:projection}, $\out$ is a play of $\game$, and by
  construction it is consistent with $\vec{\sigma}_C$: the joint action taken
  at round $t$ is $\vec{\sigma}_C(\out[0]\cdots\out[t])$. It remains to
  exhibit a counter-strategy realising the complement's moves. Since $\game$
  and $\vec{\sigma}_C$ are deterministic, the histories
  $h^{t} = \out[0]\cdots\out[t]$ are pairwise distinct, having distinct
  lengths; define $\vec{\sigma}_{-C}(h^{t}) = \beta^{t}$, where
  $\beta^{t} \in \vec{\Ac}_{-C}(\out[t])$ is any profile witnessing the edge
  $((\out[t], \vec{\sigma}_C(h^{t})), \out[t{+}1]) \in E$, and let
  $\vec{\sigma}_{-C}$ be arbitrary on all other histories. Then
  $\out_{s}(\vec{\sigma}_C, \vec{\sigma}_{-C}) = \out$, so by hypothesis
  $\out \models \psi \wedge \Lambda$ and hence
  $\bar{\rho} \in \mathrm{Obj}$. As $\sigma_2$ was arbitrary, $\sigma_1$ is
  winning.

  \smallskip\noindent\emph{From $\Pi$ to $\game$.}
  Let $\sigma_1$ be a winning player~$1$ strategy in $\Pi$ from $(s,q^{0})$.
  By \autoref{prop:projection}, every history $h = s^{0}\cdots s^{t}$ of $\game$
  from $s$ has at most one preimage in $\Pi$ consistent with
  $\sigma_1$ and ending in a $U_1$-vertex; write $\pre_{\sigma_1}(h)$ for
  that preimage when it exists. Explicitly,
  $\pre_{\sigma_1}(s^{0}) = (s^{0},q^{0})$ and
  \[
    \pre_{\sigma_1}(h \cdot s^{t+1}) \;=\;
      \pre_{\sigma_1}(h)\cdot \sigma_1(\pre_{\sigma_1}(h)) \cdot
      (s^{t+1}, \delta(q, \lab(s^{t}))) ,
  \]
  where $q$ is the automaton component of the last vertex of
  $\pre_{\sigma_1}(h)$; this is defined exactly when $s^{t+1}$ is a
  $\tr$-successor of $s^{t}$ under some completion of the joint action chosen
  by $\sigma_1$. Now set $\vec{\sigma}_C(h) = \alpha$, where
  $\sigma_1(\pre_{\sigma_1}(h)) = ((\mathit{last}(h), \alpha), q)$, taking
  $\vec{\sigma}_C(h)$ arbitrary on histories with no preimage.

  Let $\vec{\sigma}_{-C}$ be any counter-strategy and put
  $\out = \out_{s}(\vec{\sigma}_C, \vec{\sigma}_{-C})$. Every prefix of
  $\out$ has a preimage: at round $t$ the coalition plays
  $\alpha^{t} = \vec{\sigma}_C(h^{t})$, the complement plays some
  $\beta^{t}$, and $\tr(\out[t],(\alpha^{t},\beta^{t})) = \out[t{+}1]$
  witnesses the player~$2$ edge in $\Pi$. Hence there is a player~$2$
  strategy $\sigma_2$ whose play against $\sigma_1$ is $\pre_{\sigma_1}$
  applied to the prefixes of $\out$. Since $\sigma_1$ is winning, that play
  lies in $\mathrm{Obj}$, so $\out \models \psi \wedge \Lambda$ by
  \autoref{lem:projection}. As $\vec{\sigma}_{-C}$ was arbitrary,
  $\vec{\sigma}_C$ witnesses $\game, s \models \coal{C}_{\Lambda}\psi$.
\end{proof}

\begin{proof}[Proof of \autoref{thm:fm-correspondence}, items (1) and (2)]
  (1) The strategy $\sigma_1$ built in the proof of \autoref{thm:seq}
  consults $\vec{\sigma}_C$ on the projected history and otherwise needs only
  the current vertex, so a transducer of size $m$ for $\vec{\sigma}_C$ yields
  one of size $m$ for $\sigma_1$: the memory update reads the
  $\St$-component of the current $U_1$-vertex, which is exactly the letter
  $\vec{\sigma}_C$ would read.

  (2) In the second half of the proof of \autoref{thm:seq}, the joint action
  $\vec{\sigma}_C(h)$ is determined by the memory state of $\sigma_1$ after
  $\pre_{\sigma_1}(h)$ together with the automaton component reached. The
  latter is not a function of $\mathit{last}(h)$, so $\vec{\sigma}_C$ must
  track $\dpa_{\psi}$ itself; the product of the two gives a transducer of
  size $m\cdot|Q|$. Each agent of $C$ can implement a private copy of the transducer.

\end{proof}

\begin{proof}[Proof of \autoref{thm:fm-correspondence}, item (3)]
  Suppose some counter-strategy $\vec{\sigma}_{-C}$ produces an outcome
  $\out$ with $\out \not\models \psi\wedge\Lambda$. Let $\sigma_1$ be the
  player~$1$ strategy of \autoref{thm:fm-correspondence}(1), of memory size $m$, and
  let $\Pi[\sigma_1]$ be the finite graph obtained from $\Pi$ by fixing
  $\sigma_1$; its vertices are pairs of a $\Pi$-vertex and a memory state,
  so $|\Pi[\sigma_1]| \leq m\cdot|U_1 \cup U_2|$. The preimage of $\out$ is a
  path in $\Pi[\sigma_1]$ violating $\mathrm{Par}\wedge\Lambda$, that is,
  satisfying $\neg\mathrm{Par} \vee \neg\Lambda$. Each disjunct admits an
  ultimately periodic witness in a finite graph. For $\neg\mathrm{Par}$ this
  holds because the condition is $\omega$-regular. For $\neg\Lambda$, choose
  a conjunct $\mpf_j \geq q_j$ of $\Lambda$ that this path violates;
  if every reachable cycle of $\Pi[\sigma_1]$ had mean at least $q_j$ in
  dimension $j$, then decomposing each finite prefix into cycles and a
  residual simple path would give every infinite path $\liminf$ mean at
  least $q_j$; so such a path has $\liminf$ mean below $q_j$ only if some
  reachable cycle of $\Pi[\sigma_1]$ has mean below $q_j$ in that dimension,
  and traversing that cycle forever is an ultimately periodic witness. Fix
  such a lasso in
  $\Pi[\sigma_1]$; it is realised by a player~$2$ strategy of memory size at
  most $|\Pi[\sigma_1]|$. Translating that strategy back into a
  counter-strategy in $\game$ requires computing, at each round, the joint
  action $\vec{\sigma}_C$ prescribes; since $\vec{\sigma}_C$ has memory
  $m$, a counter-strategy of memory at most $m\cdot|\Pi[\sigma_1]|$
  suffices. This is a finite-memory counter-strategy defeating
  $\vec{\sigma}_C$, contrary to hypothesis.
\end{proof}

\begin{proof}[Proof of the final claim of \autoref{thm:fm-correspondence}]
  Item~(3) concerns a single strategic formula with
  a fixed $\LTL$ path formula, so the passage to arbitrary $\varphi$ requires
  an induction. We argue by structural induction on $\varphi$, showing that the
  two relations give it the same extension $\ext{\varphi} \subseteq \St$.

  The atomic and Boolean cases are immediate, since neither relation
  quantifies over strategies there. Let $\varphi = \coal{C}_{\Lambda}\psi$
  and let $\varphi_1,\dots,\varphi_k$ be the maximal state subformulae of
  $\psi$. By the induction hypothesis each $\ext{\varphi_i}$ is the same
  under both relations, so the extended labelling, and hence the $\LTL$
  formula $\psi'$ of \autoref{lem:subst}, are the same under both. By that
  lemma, satisfaction of $\psi$ along a play is then the fixed
  $\omega$-regular condition $\lab(\out) \models \psi'$, independent of the
  strategy classes. Applying item~(3) of \autoref{thm:fm-correspondence} to $\psi'$
  gives that a finite-memory $C$-strategy defeats every finite-memory
  counter-strategy if and only if it defeats every counter-strategy, so the
  two relations agree at $\varphi$. For the second claim, combine this with
  \autoref{thm:fm-correspondence}(1) and~(2).
\end{proof}

\section{Finite-memory mean-payoff parity games}
\label{app:energy}

\begin{proof}[Proof of \autoref{lem:fm-energy}]
  ($\Leftarrow$) A strategy winning the energy parity game with initial
  credit $c_0$ keeps every partial sum of weights at least $-c_0$ and
  satisfies the parity condition, so along any consistent play
  $\frac{1}{n}\sum_{t<n}w(v^{t}) \geq -c_0/n \to 0$ and hence
  $\mpf \geq 0$. Player~$1$ has such a strategy of memory size
  $O(|V|\cdot c\cdot W)$ \cite{ChatterjeeD12}.

  ($\Rightarrow$) Let $\sigma_1$ be a finite-memory strategy of size $m$
  winning $\mathrm{Par}\wedge(\mpf\geq0)$, and let $H$ be the finite graph
  $\game^{\bullet}[\sigma_1]$ obtained by fixing $\sigma_1$, with at most
  $|V|\cdot m$ vertices. Every reachable cycle of $H$ has non-negative
  weight sum: a cycle of negative sum could be traversed forever by
  player~$2$, yielding a consistent play of mean payoff equal to that cycle's
  mean, which is negative, contradicting the hypothesis. Hence along any path
  of $H$ the running sum never drops below $-|V|\cdot m\cdot W$, since a
  prefix decomposes into simple cycles, each of non-negative sum, and a
  simple path of length at most $|V|\cdot m$. Hence $\sigma_1$ wins the energy
  parity game with initial credit $|V|\cdot m \cdot W$.

  Membership in $\NP\cap\coNP$ and the memory bound are then those of energy
  parity games \cite{ChatterjeeD12}.
\end{proof}

\section{Direct constructions for restricted fragments}
\label{app:frag}

\begin{proof}[Proof of \autoref{lem:atl-constructions}]
  We use throughout that $\Lambda$ is prefix-independent.

  \emph{Next.} Solve the one-priority mean-payoff game $\game^{C}$ for the
  threshold of $\Lambda$ and let $R$ be player~$1$'s winning region. Then
  $\game, s \models \coal{C}_{\Lambda}\Xop\varphi$ iff there is
  $\alpha \in \vec{\Ac}_C(s)$ such that $(s,\alpha) \in R$ and every $s'$
  with $((s,\alpha),s') \in E$ satisfies $\varphi$. Indeed, such an $\alpha$
  lets the coalition force the next state into $\ext{\varphi}$ while
  retaining a strategy enforcing $\Lambda$ from $(s,\alpha)$ onwards, and
  conversely any witness prescribes such an $\alpha$ at $s$.

  \emph{Safety.} Compute by attractor computation the maximal region $R'$
  from which player~$1$ can keep the play inside $\ext{\varphi}$ forever in
  $\game^{C}$; this is a safety game, solvable in linear time. Restrict
  $\game^{C}$ to $R'$, removing player~$1$ moves that leave it. Any strategy
  enforcing $\Gop\varphi$ stays in $R'$, and any strategy staying in $R'$
  enforces $\Gop\varphi$; since $\Lambda$ is prefix-independent it is
  unaffected by the restriction. The result is a one-priority mean-payoff
  game.

  \emph{Until.} Take two copies of $\game^{C}$ indexed by a bit recording
  whether a $\varphi_2$-state has already been visited, together with a
  self-looping sink. Throughout, a vertex of $\game^{C}$ is said to satisfy a
  state formula when its $\St$-component does, which lifts
  $\ext{\varphi_1}$ and $\ext{\varphi_2}$ to the intermediate vertices as
  well.

  The vertices are $V \times \{0,1\}$ together with $\bot$, and the initial
  vertex for a state $s$ is
  \[
    \iota(s) \;=\;
    \begin{cases}
      (s, 1) & \text{if } s \in \ext{\varphi_2}, \\
      (s, 0) & \text{if } s \in \ext{\varphi_1}\setminus\ext{\varphi_2},\\
      \bot   & \text{otherwise.}
    \end{cases}
  \]
  The case distinction is what makes the construction correct when
  $\varphi_2$ already holds at $s$, so that $\varphi_1 \Uop \varphi_2$ is
  satisfied with $k = 0$. From $(v,0)$ the play moves to $(v',1)$ if $v'$
  satisfies $\varphi_2$, to $(v',0)$ if $v'$ satisfies $\varphi_1$ but not
  $\varphi_2$, and to $\bot$ otherwise; from $(v,1)$ it moves to $(v',1)$ for
  every successor $v'$. Priorities are
  $\bar{\pri}(v,0) = \bar{\pri}(\bot) = 1$ and $\bar{\pri}(v,1) = 0$, and
  weights are inherited from the $\St$-component, with $\bot$ given weight
  $0$.

  A play satisfies the parity condition iff it eventually enters the
  $1$-indexed copy, that is, iff its projection satisfies
  $\varphi_1 \Uop \varphi_2$: a play remaining in the $0$-indexed copy or
  entering $\bot$ has least priority $1$ occurring infinitely often and is
  losing, while a play entering the $1$-indexed copy has priority $0$
  recurring. Note that the parity condition alone does the work; no
  additional penalty on $\bot$ is needed, and none is imposed. Since
  $\Lambda$ is prefix-independent and the tail of a winning play lies
  entirely in the $1$-indexed copy, the mean payoff of a play equals that of
  its projection. The arena has $2|V| + 1$ vertices and two priorities.
\end{proof}

\begin{proof}[Proof of \autoref{lem:gr1-dpa}]
  Write $A = \bigwedge_{\ell\leq m}\Gop\Fop\psi_{\ell}$ and
  $B = \bigwedge_{r\leq k}\Gop\Fop\vartheta_{r}$, so
  $\theta = \neg A \vee B$. A counter $\ell \in \{1,\dots,m\}$ is incremented
  modulo $m$ whenever $\psi_{\ell}$ holds, and the event $a$ is its
  wrap-around; then $A$ holds iff $a$ occurs infinitely often. The
  guarantees are treated the same way, with counter $r$ and event $b$.

  The events $a$ and $b$ are produced by transitions, whereas
  \autoref{def:rpp} reads priorities off states, so the event is recorded in
  the state: a control state is a triple $(\ell, r, e)$ with
  $e \in \{\mathsf{none}, a, b, ab\}$ recording which wrap-arounds were
  caused by the letter just read. This multiplies the state count by four and
  leaves the bound $O(mk)$ unchanged. Assign priority $0$ to states with
  $e \in \{b, ab\}$, priority $1$ to states with $e = a$, and priority $2$ to
  states with $e = \mathsf{none}$.

  The transition function is not tabulated over the alphabet $2^{\AP}$, which
  would be exponential in $|\AP|$. It is kept symbolic: each transition out
  of a control state is guarded by one of the Boolean formulae $\psi_{\ell}$
  or $\vartheta_{r}$, and in the product of \autoref{def:rpp} the two
  counters are updated by evaluating those formulae on $\lab(s)$, in time
  polynomial in $\size{\theta}$ per transition.
\end{proof}

\end{document}